\documentclass[11pt,letterpaper]{article}

\usepackage[margin=0.95in]{geometry}
\usepackage{amsmath,amssymb,amsthm,mathtools,bm}
\usepackage{algorithm}
\usepackage{algpseudocode}
\usepackage{booktabs}
\usepackage{enumitem}
\usepackage{booktabs,array,longtable,tabularx}
\usepackage{graphicx}
\usepackage{float}
\usepackage[all,cmtip]{xy}
\usepackage{qcircuit}
\usepackage{xcolor}
\usepackage{mdframed}
\usepackage{tikz}
\usetikzlibrary{arrows.meta,calc}
\usepackage{microtype}
\usepackage[colorlinks=true,linkcolor=blue!55!black,citecolor=red!55!black,urlcolor=blue!55!black]{hyperref}
\usepackage{authblk}
\allowdisplaybreaks
\setlist{leftmargin=*,nosep}
\newcolumntype{Y}{>{\raggedright\arraybackslash}X}
\newcolumntype{Z}{>{\centering\arraybackslash}X}

\theoremstyle{plain}
\newtheorem{theorem}{Theorem}[section]
\newtheorem{lemma}[theorem]{Lemma}
\newtheorem{proposition}[theorem]{Proposition}
\newtheorem{corollary}[theorem]{Corollary}
\newtheorem{property}[theorem]{Property}
\theoremstyle{definition}
\newtheorem{definition}[theorem]{Definition}

\theoremstyle{remark}

\newcommand{\F}{\mathbb{F}}

\newcommand{\im}{\operatorname{im}}
\newcommand{\rank}{\operatorname{rank}}
\newcommand{\rad}{\operatorname{rad}}
\newcommand{\supp}{\operatorname{supp}}

\newcommand{\Tr}{\operatorname{Tr}}

\newcommand{\Span}{\operatorname{span}}

\newcommand{\TFD}{\operatorname{TFD}}
\newcommand{\Bell}{\operatorname{Bell}}

\newcommand{\ket}[1]{\left|#1\right\rangle}

\newcommand{\norm}[1]{\left\lVert #1\right\rVert}

\newcommand{\floor}[1]{\left\lfloor #1\right\rfloor}
\newcommand{\ceil}[1]{\left\lceil #1\right\rceil}

\newcommand{\details}[2]{\hyperref[#1]{#2}}

\definecolor{revisionyellow}{RGB}{255,242,166}
\newif\ifshowyellowrevisions
\showyellowrevisionstrue

\title{Quantum Gibbs State Preparation via Relative Decoding:
From Fast-Mixing Sources to Broader Target Classes}

\author[1]{Zhong-Xia Shang%
  \thanks{\texttt{zhongxia.shang@math.ku.dk}}}

\author[1]{Yufei Wang%
  \thanks{\texttt{yuw@math.ku.dk}}}
  
\author[1]{Daniel Stilck Fran\c{c}a%
  \thanks{\texttt{dsfranca@math.ku.dk}}}

\affil[1]{Department of Mathematical Sciences,
University of Copenhagen, Denmark}

\date{}

\begin{document}
\hypersetup{pageanchor=false}

\pagenumbering{gobble}

\begin{titlepage}
\maketitle
\centering
\vspace{0.7em}

\begin{minipage}{0.93\textwidth}

\begin{abstract}
Provable guarantees for preparing quantum Gibbs states, such as rapid
mixing or gapped coherent preparation paths, are known only for
restricted Hamiltonian families and usually must be re-derived for each
new family. Building on homomorphic polynomial transduction, a generalization of decoded quantum interferometry (DQI) and Hamiltonian DQI (HDQI), we show how
relative decoding transfers such a guarantee from a source Hamiltonian
$H_A$ to a target $H_B$. Writing each Hamiltonian as a sum of $m$ Pauli terms, the subsets of
terms whose product is proportional to the identity, its relations,
form a binary linear code, $K_A$ for the source and $K_B$ for the
target, with the Pauli-label matrix as parity-check matrix as in DQI. Starting from the canonical thermofield
double (TFD) of $H_A$, a Bell transform, a reversible label map, and a
coherent decoder for the quotient code $K_B/K_A$ prepare an approximate TFD of $H_B$, and
hence its Gibbs state. Because this decoder only resolves target
relations absent from the source, the reachable inverse temperature is
set by the relative distance $d_{\rm rel}$, the fewest terms in any such
relation. It can far exceed the ordinary distance $d_{\rm ord}$, the
fewest terms in any target relation, which bounds the uniform exact decoding
radius of DQI and HDQI. We show that linear
$d_{\rm rel}$ with efficient decoding at linear radius certifies a
constant inverse temperature. As an example, starting from a
nearest-neighbor spin chain whose TFD is known to be preparable at every finite
temperature, a sparse classical parity-check matrix yields
bounded-degree, geometrically nonlocal, noncommuting targets with
$d_{\rm ord}=3$ and $d_{\rm rel}=\Theta(m)$. To our knowledge, this gives a new class of efficiently preparable TFDs.

\end{abstract}

\end{minipage}
\vfill
\end{titlepage}
\clearpage
\thispagestyle{empty}
\begingroup
\small
\setcounter{tocdepth}{1}
\tableofcontents
\endgroup
\clearpage
\hypersetup{pageanchor=true}
\pagenumbering{arabic}

\section{Introduction}
\label{sec:introduction}

Preparing the Gibbs state $\rho_\beta(H)=e^{-\beta H}/\Tr\,e^{-\beta H}$ of a
many-body Hamiltonian is a central task in quantum simulation and one of the
most natural candidates for a quantum advantage. Two routes dominate the
recent literature. Quantum Gibbs
samplers~\cite{ChenEtAlNature2025,ChenKastoryanoGilyen} run efficiently
implementable dynamics whose fixed point is $\rho_\beta(H)$, and coherent
versions of the same constructions~\cite{ChenKastoryanoGilyen,RouzeFrancaAlhambra}
prepare the thermofield double (TFD), the canonical purification of
$\rho_\beta(H)$. Their cost is governed by a mixing time or a spectral gap,
and provable bounds exist only for restricted families, such as local
Hamiltonians at high temperature~\cite{RouzeFrancaAlhambra} or finite-range
spin chains at every fixed
temperature~\cite{BergamaschiChen1D,Kim2026RapidMixing1D}. Every new family
requires a new analysis.

Decoded quantum interferometry (DQI)~\cite{JordanEtAlDQI} and its
Hamiltonian version HDQI~\cite{HamiltonianDQI} take a different route. They
prepare states whose amplitudes are a polynomial $Q_D$ of the Hamiltonian by
combining Fourier interference with coherent decoding of a classical code
attached to the Hamiltonian; choosing $Q_D\approx e^{-\beta x/2}$ gives an
approximate TFD. The certified inverse temperature is set by the polynomial
approximation and the decoding radius: the standard exact-decoding guarantee
for a degree-$D$ filter is $2D<d_{\rm ord}$, where $d_{\rm ord}$ is the
distance of the code.

\paragraph{This work.}
We combine the two routes. A Gibbs-preparation guarantee for a source
Hamiltonian $H_A$ supplies its TFD; a relative-decoding circuit then converts
this state into the TFD of a target Hamiltonian $H_B$ with different, possibly
geometrically nonlocal, interactions. The general mechanism, homomorphic
polynomial transduction, was introduced in our companion
paper~\cite{ShangFranca2026Transduction} for general operator systems,
including fermionic, qudit and truncated bosonic ones. Here we specialize it
to Pauli Hamiltonians, where compatibility reduces to three checkable
linear-algebraic conditions and the decoding task to a quotient code, and we
work out the thermal application in detail: an explicit circuit, a
trace-distance error bound, a degree-to-temperature certificate built on the
polynomial approximation of~\cite{ShangFranca2026Transduction}, structural
results on which source--target pairs are possible, and a concrete family of
targets with an efficient decoder.

\paragraph{Relations, codes and relative decoding.}
Write source and target as sums of $m$ matched Pauli terms,
$H_A=\sum_i\lambda_i\alpha_iP_A(a_i)$ and $H_B=\sum_i\lambda_i\beta_iP_B(b_i)$,
with $\lambda_i>0$, signs $\alpha_i,\beta_i\in\{\pm1\}$, and binary Pauli
labels collected into matrices $A=[a_1,\ldots,a_m]$ and $B=[b_1,\ldots,b_m]$.
A \emph{relation} of $H_B$ is a subset of its terms whose product is
proportional to the identity, for instance $Z_1$, $Z_2$ and $Z_1Z_2$. The
indicator vectors of relations form a binary linear code $K_B=\ker B$ with
parity-check matrix $B$, and likewise $K_A=\ker A$ for the source. Relations
are precisely what a decoder has to resolve. Expanding $Q_D(H_B)$ gives a sum
over subsets $z\in\F_2^m$ of at most $D$ terms; two subsets produce the same
Pauli label $Bz$ exactly when they differ by a relation; and recovering $z$
from $Bz$ is syndrome decoding, unambiguous only if
$2D<d_{\rm ord}:=\min_{0\ne c\in K_B}|c|$. This is the bottleneck of DQI and
HDQI.

The observation behind this paper is that $z$ itself is never needed. The
transduction circuit only has to erase the \emph{source} label $Az$, which is
already blind to source relations. When $K_A\subseteq K_B$, and the
anticommutation patterns and relation phases of the two Hamiltonians agree,
$Bz$ determines $Az$ for all $|z|\le D$ as soon as
\begin{equation*}
    2D<d_{\rm rel}:=\min_{c\in K_B\setminus K_A}|c| ,
\end{equation*}
the distance of the quotient code $K_B/K_A$. Always $d_{\rm rel}\ge d_{\rm ord}$,
with strict inequality exactly when every minimum-weight target relation is
already a source relation. Absorbing the short target relations can therefore
enlarge the decoding window and, with an efficient relative decoder, certify
lower temperatures. DQI and HDQI are the relation-free case $K_A=0$.

\paragraph{Main results.}
The first result gives the preparation error and temperature range in terms
of the available relative-decoding degree.

\begin{theorem}[Informal; see Theorems~\ref{thm:efficient-poly}
and~\ref{thm:tfd-transduction}]
Let $H_A$ and $H_B$ have the same anticommutation pattern, $K_A\subseteq K_B$,
and matching relation phases, and let a coherent relative decoder for
$K_B/K_A$ be efficient up to degree $D_{\rm dec}<d_{\rm rel}/2$. Given the
canonical TFD of $H_A$ at inverse temperature $\beta$ to trace distance
$\delta_{\rm src}$, an efficient circuit prepares the TFD of $H_B$, and hence
its Gibbs state, to trace distance $\delta_{\rm src}+2\eta+\delta_{\rm circ}$
whenever
\[
    D_{\rm dec}\ \ge\ \beta\Lambda
    +O\bigl(\sqrt{\beta\Lambda\ln(1/\eta)}+\ln(1/\eta)\bigr),
    \qquad \Lambda\ge\|H_A\| .
\]
In particular, $D_{\rm dec}=\Omega(m)$ and $\Lambda=O(m)$ certify every
inverse temperature below a fixed constant.
\end{theorem}

For an extensive spectral bound, efficient relative decoding at linear
radius certifies a constant-inverse-temperature preparation window.
Our main example realizes these conditions using a classical code and a
one-dimensional source with an efficiently preparable TFD.

\begin{theorem}[Informal; see Theorem~\ref{thm:1d-source-transduction}]
Fix $N$ and real coefficients $\lambda_i^{\mu}$ with $0<|\lambda_i^{\mu}|\le J_0$.
Let the source be the nearest-neighbor chain on $N$ two-qubit cells
$(q_i,c_i)$,
\[
    H_A=\sum_{i=1}^{N}\bigl(
        \lambda_i^xX_{q_i}+\lambda_i^zZ_{q_i}
        +\lambda_i^cZ_{c_i}+\lambda_i^{cq}Z_{c_i}Z_{q_i}\bigr)
    +\sum_{i=1}^{N-1}\lambda_i^{zz}Z_{q_i}Z_{q_{i+1}},
\]
whose canonical TFD is efficiently preparable at every fixed
temperature~\cite{BergamaschiChen1D}. Let $F\in\F_2^{(N-k)\times N}$ be a
parity-check matrix with columns $f_1,\ldots,f_N$, distance
$d_F\ge\delta N$, and an efficient syndrome decoder that recovers every
error of weight at most $\lfloor\tau N\rfloor$, where $0<2\tau<\delta$ are
constants. Replace the $N$ central qubits by $N-k$ syndrome qubits
$R_1,\ldots,R_{N-k}$ and each $Z_{c_i}$ by the parity check
$Z_R(f_i):=\prod_{a:(f_i)_a=1}Z_{R_a}$, so that
\[
    H_B=\sum_{i=1}^{N}\bigl(
        \lambda_i^xX_{q_i}+\lambda_i^zZ_{q_i}
        +\lambda_i^cZ_R(f_i)+\lambda_i^{cq}Z_R(f_i)Z_{q_i}\bigr)
    +\sum_{i=1}^{N-1}\lambda_i^{zz}Z_{q_i}Z_{q_{i+1}}.
\]
Then $H_A$ and $H_B$ are compatible, with $m=5N-1$ terms,
\[
    d_{\rm ord}=3,\qquad d_{\rm rel}=d_F=\Theta(m),
    \qquad D_{\rm dec}=\lfloor\tau N\rfloor ,
\]
and relative decoding is exactly syndrome decoding of $F$. The TFD and Gibbs
state of $H_B$ are prepared in polynomial time for every fixed
$\beta<\beta_0=\Theta(\tau/J_0)$. If $F$ has bounded row and column weights,
$H_B$ has bounded-weight terms and bounded qubit degree; if the supports of
the columns of $F$ have unbounded diameter, its interactions are
geometrically nonlocal.
\end{theorem}
To our knowledge this yields a
new class of noncommuting Hamiltonians whose TFD, and hence Gibbs state, is
efficiently preparable at constant inverse temperature.
Physically, $H_A$ is a mixed-field Ising chain on the spins $q_i$, with
transverse field $\lambda_i^x$, longitudinal field $\lambda_i^z$ and coupling
$\lambda_i^{zz}$, in which each longitudinal field is shifted by
$\lambda_i^{cq}$ times a classical spin $Z_{c_i}$. In $H_B$ that classical
spin is replaced by a parity $Z_R(f_i)$ of the syndrome spins: the quantum
sector is unchanged and remains one-dimensional, while the classical sector
becomes a parity-check spin system whose checks drive the longitudinal
fields of the chain.

Relative decoding is what makes this work. Every cell carries the
weight-three relation $Z_{q_i}\cdot Z_{c_i}\cdot Z_{c_i}Z_{q_i}=I$ and every
bond the relation $Z_{q_i}\cdot Z_{q_{i+1}}\cdot Z_{q_i}Z_{q_{i+1}}=I$, so
$d_{\rm ord}=3$ and the ordinary exact-decoding guarantee for HDQI on $H_B$
reaches only degree one. Also vHDQI is limited to only uniform coefficients. The chain $H_A$ already
contains all of these short relations, and its TFD is efficiently
preparable. The quotient code $K_B/K_A$ then consists exactly of the
codewords of $F$, so $d_{\rm rel}=d_F$, and the coherent relative decoder is
the classical decoder for $F$ run reversibly.

\paragraph{Further results.}
For commuting relation-free sources, the setting of DQI for Max-XOR-SAT, the
polynomial-filtered Choi state is Clifford-equivalent to a matrix-product
state of bond dimension $D+1$
(Theorem~\ref{thm:commuting-relation-free-source}); this handles nonuniform
coefficients directly and recovers, by a shorter argument, an observation
of~\cite{GeneralPauliHDQI}. Section~\ref{sec:properties} collects structural
properties of the construction: for a fixed anticommutation pattern $\Sigma$,
exactly the chains $K_A\subseteq K_B\subseteq\ker\Sigma$ are realizable; the
largest achievable $d_{\rm rel}$ for a fixed source or a fixed target has an
exact coset description; finite $d_{\rm rel}$ rules out termwise unitary
conjugation; a minimum-weight ordinary decoder gives a valid relative
decoder throughout the relative window; and the kernel dimension and
anticommutation rank bound the number of interaction terms in reduced
Pauli families. Section~\ref{sec:discussion} identifies sampling separations
and comparisons with target dynamics as further applications to pursue.

\paragraph{Organization.}
Section~\ref{sec:related-work} reviews prior work, including the precise
relation to the companion paper. Section~\ref{sec:quotient-decoding-main}
defines quotient codes and relative decoding, and Section~\ref{sec:framework}
presents the transduction framework and circuit.
Section~\ref{sec:dqi-hdqi-special} recovers DQI and HDQI as relation-free
special cases, and Section~\ref{sec:tfd-main} proves the thermal guarantees.
Sections~\ref{sec:examples1} and~\ref{sec:examples2} treat commuting
relation-free sources and the one-dimensional construction.
Section~\ref{sec:properties} presents structural properties, and
Section~\ref{sec:discussion} discusses open directions. Proofs are in the
appendices.

\section{Relation to previous work}
\label{sec:related-work}
Quantum Gibbs sampling has received considerable attention over the recent years, as we will see below. In this section, we will both collect constructions of TFD states that can be used as input to our algorithm, complementary methods for Gibbs sampling and other DQI-related results.
\paragraph{Thermofield-double preparation.}
Our framework takes an efficiently prepared canonical TFD as input.
Preparation methods include energy filtering and imaginary-time
evolution~\cite{PoulinWocjan2009,ChowdhurySomma2017,AnChildsLin2026Nonunitary},
quantum simulated annealing based on a quantum--quantum Metropolis
walk~\cite{YungAspuruGuzik2012QuantumMetropolis}, and Szegedy quantization
of quantum maps satisfying GNS detailed balance~\cite{WocjanTemme2023Szegedy}.
Parent-Hamiltonian approaches include exact TFD constructions and their
approximate realizations~\cite{CottrellEtAl2019}, as well as parents derived
from detailed-balance dynamics~\cite{ChenEtAlNature2025,ChenKastoryanoGilyen}.
Ge, Moln{\'a}r, and Cirac give adiabatic preparation of thermal purifications
for commuting lattice Hamiltonians with near-linear gate count and
polylogarithmic depth under a uniform path-gap assumption, together with a
polynomial-time extension to noncommuting systems at high
temperature~\cite{GeMolnarCirac2016Adiabatic}.
Recent developments include detectability-lemma annealing for commuting
Hamiltonians~\cite{FangLuTongZhao2026Detectability} and walk-free filtering
with inverse-square-root gap dependence from constant-overlap warm starts, including
explicit noncommuting examples~\cite{LengJiangLin2026WalkFree}.
Other constructions use weakly interacting fermionic
parents~\cite{SmidMeisterBertaBondesan2025FermiHubbard}, modular
annihilators~\cite{YiTakahashiZhou2026ModularAnnihilator}, or double-bracket
imaginary-time evolution~\cite{Wright2026DoubleBracketTFD}.
Their costs are controlled by spectral gaps, coherent access, initial
overlap, or normalization. Another approach is through fluctuation-theorem algorithms that transform a
simpler thermal purification into the desired TFD~\cite{HolmesEtAl2022Fluctuation}.

Rigorous efficient canonical-TFD regimes include high-temperature local
Hamiltonians~\cite{RouzeFrancaAlhambra} and finite-range spin chains at
every fixed finite inverse temperature~\cite{BergamaschiChen1D}.
We use the latter's canonical-purification corollary, which gives
polylogarithmic depth and near-linear size with temperature-dependent
constants. 

\paragraph{DQI and its extensions.}
DQI combines Fourier interference with coherent decoding~\cite{JordanEtAlDQI},
following a mechanism related to Regev's reduction~\cite{RegevLWE}; HDQI
extends it to noncommuting Pauli Hamiltonians~\cite{HamiltonianDQI}.
Extensions address arbitrary coefficients~\cite{GeneralPauliHDQI},
multivariate weighted constraints~\cite{BuGuLi2026MultivariateDQI}, and
structured matrix-product pilot states~\cite{WocjanPilot}.
Further work treats algebraic-geometry codes~\cite{GuJordan2025AlgebraicGeometry,MatthewsShapiro2026COFI}
and non-Hamming metrics~\cite{KrajenbrinkEtAl2026BeyondHamming}. Decoding-based approaches to Gibbs preparation include Refs~\cite{HastingsSchmidhuber2026RedundancyFree,PiveteauQuekFrancoGarrido2026QuantumDecoding,ZhouEtAl2026ShatteredGibbs}. To the best of our knowledge, our work and Ref.~\cite{ShangFranca2026Transduction} are the first to exploit relative decoding in the context of DQI and can directly combine existing TFD preparation methods. Several works also have studied the limitations of DQI-like approaches to prepare interesting states or states that achieve a quantum advantage. Simulation complexity, conditional spin-glass obstructions, and numerical
MCMC benchmarks are studied in
Refs.~\cite{MarwahaDQIComplexity,AnschuetzEtAl2025SpinGlass,GilFusterEtAl2026MCMCDQI}. Other advantages hold in membership-oracle or streaming-space
models~\cite{KramerEtAl2026OracleDQI,WuYang2026StreamingDQI,BoulandDingJain2026RandomStreams}. To the best of our knowledge, the examples covered here are not readily covered by any of these works.

\paragraph{Relation to the companion framework.}
Shang and Stilck Fran{\c{c}}a~\cite{ShangFranca2026Transduction} introduce
homomorphic polynomial transduction and the transfer of a supplied source
TFD through relative decoding for general operator systems. We specialize
the construction to Pauli Hamiltonians and give an explicit noncommuting
chain application. A classical syndrome decoder and an efficiently
prepared source TFD yield geometrically nonlocal targets with
$d_{\rm ord}=3$, $d_{\rm rel}=\Theta(m)$, and a constant-inverse-temperature
preparation window. We also give a thorough exploration of various properties of the framework.

\section{Quotient codes and relative decoding}
\label{sec:quotient-decoding-main}

We start by recalling some basic notions of quotient and relative codes~\cite{HuffmanPless2003}.
The decoding layer can be isolated from the Hamiltonian realization.  Let
\begin{equation}
 A\in\F_2^{r_A\times m},\qquad B\in\F_2^{r_B\times m},
 \qquad K_A:=\ker A,\quad K_B:=\ker B,
 \label{eq:main-code-matrices}
\end{equation}
The label $u=Az$ with $z\in\F_2^m$ identifies the coset $z+K_A$.  If $K_A\subseteq K_B$, then $Bz$ depends only on this coset and $B$ descends to
\begin{equation}
 \overline B:\F_2^m/K_A\to\im B,
 \qquad \overline B(z+K_A)=Bz,
 \qquad \ker\overline B=K_B/K_A.
 \label{eq:main-quotient-map}
\end{equation}
Equivalently, under $\F_2^m/K_A\cong\im A$, the induced map is
\begin{equation}
 M:\im A\to\im B,\qquad M(Az)=Bz.
 \label{eq:main-induced-map}
\end{equation}
Thus $K_B/K_A$ is exactly the ambiguity that remains after source relations have been identified.

For a nonnegative integer $D$ and $u\in\im A$, define
\begin{equation}
 w_A(u):=\min_{z:Az=u}|z|,
 \qquad
 \mathcal S_D(A):=\{u:w_A(u)\le D\},
 \qquad
 d_{A\to B}:=\min_{\substack{0\ne u\in\ker M}}w_A(u)
              =\min_{c\in K_B\setminus K_A}|c|.
 \label{eq:main-source-weight-support}
\end{equation}
We set $d_{A\to B}=+\infty$ when $K_A=K_B$.  The fundamental code fact is
\begin{equation}
 M|_{\mathcal S_D(A)}\text{ is injective}
 \quad\Longleftrightarrow\quad
 2D<d_{A\to B}.
 \label{eq:main-relative-injectivity}
\end{equation}

\paragraph{Standard versus quotient-aware decoding.}
An ordinary decoder, on the promise $|z|\le D$, recovers the full representative
\begin{equation}
 \operatorname{Dec}_B(Bz)=z,
\end{equation}
whereas relative decoding only requires
\begin{equation}
 \operatorname{Dec}_{A\to B}(Bz)=Az.
 \label{eq:main-relative-decoder-goal}
\end{equation}
Let $d_B:=\min_{0\ne c\in K_B}|c|$, with $d_B=+\infty$ when $K_B=0$, be the ordinary target distance. Then $d_{A\to B}\ge d_B$. When $d_B<\infty$, the inequality is strict exactly when every minimum-weight target relation already lies in $K_A$.  Appendix~\ref{app:relative} proves these statements.

Distance guarantees that the output $Az$ is unique, not that it is easy to compute.  We call the relative decoder efficient to degree $D$ when a polynomial-size quantum circuit implements $\ket y\ket0\mapsto\ket y\ket{Az}$ on the promise $y=Bz$ and $w_A(Az)\le D$.  The interesting regime here is
\begin{equation}
 d_B\le 2D<d_{A\to B}.
 \label{eq:main-relative-only-window}
\end{equation}

\section{Transduction framework and circuit}
\label{sec:framework}
\label{sec:main-results}

\noindent\textbf{Setup.}
\begin{definition}[Pauli label and Pauli operator]\label{pauli}
The binary label $g=(x,z)\in \F_2^{2n}$ represents the $n$-qubit Pauli operator
\begin{equation}
P_n(g)=P_n(x,z):=\bigotimes_{j=1}^n i^{x_jz_j}X_j^{x_j}Z_j^{z_j}.
\label{eq:pauli-representative}
\end{equation}
\end{definition}
\noindent Under this definition, for $n=1$, we have $P_1(1,0)=X$, $P_1(0,1)=Z$, and $P_1(1,1)=Y$. 

Write the source and target Hamiltonians as
\begin{equation}
H_A=\sum_{i=1}^m \lambda_i\alpha_i P_{A}(a_i),\qquad
H_B=\sum_{i=1}^m \lambda_i\beta_i P_{B}(b_i),
\label{eq:hamiltonians}
\end{equation}
where $P_A=P_{n_A}$, $P_B=P_{n_B}$, $\lambda_i>0$, $\alpha_i,\beta_i\in\{\pm1\}$, $a_i\in\F_2^{2n_A}$, and $b_i\in\F_2^{2n_B}$. We call $H_A$ the source Hamiltonian and $H_B$ the target Hamiltonian. The source and the target can act on different numbers of qubits (i.e. $n_A\neq n_B$).

Our goal is to prepare the polynomial-filtered Choi state of $H_B$ from the polynomial-filtered Choi state of $H_A$ with the definition as follows.
\begin{definition}[Polynomial-filtered Choi state ]\label{pfcs}
Given a $D$-degree polynomial 
\begin{equation}
Q_D(x)=\sum_{\ell=0}^D q_\ell x^\ell,\qquad q_\ell\in\mathbb R,
\end{equation}
Assuming $Q_D(H_L)\ne0$, the normalized polynomial-filtered Choi state of $H_L$ with $L\in\{A,B\}$ is
\begin{equation}
\ket{\Psi_L^Q}:=\frac{(Q_D(H_L)\otimes I)\ket{\Phi_{n_L}}}
{\norm{(Q_D(H_L)\otimes I)\ket{\Phi_{n_L}}}},\label{eq:filtered-choi-def}
\end{equation}
where
\begin{equation}
\ket{\Phi_n}=2^{-n/2}\sum_{x\in\F_2^n}\ket{x}\ket{x}.
\end{equation}
is the normalized maximally entangled state on $n$ physical and $n$ reference qubits.
\end{definition}
The polynomial-filtered Choi state is a purification of the polynomial-filtered Hamiltonian density matrix
\begin{equation}
\rho_L^Q:=\frac{Q_D(H_L)^2}{\Tr[Q_D(H_L)^2]}.
\label{eq:filtered-density}
\end{equation}
Choosing the polynomial allows approximate optimization and Gibbs-state preparation, the latter being the main focus here. 

\noindent\textit{\textbf{Transduction conditions.—}}$H_A$ and $H_B$ require satisfying several conditions to enable the transduction. To bring out these conditions, we need to define several notations. We define
\begin{equation}
A:=[a_1,\ldots,a_m],\qquad B:=[b_1,\ldots,b_m],
\label{eq:AB-def}
\end{equation}
as the Hamiltonian matrices of $H_A$ and $H_B$ respectively. We also use the shorthand $\ell_i^A=a_i$ and $\ell_i^B=b_i$. Let
\begin{equation}
J_n=\begin{pmatrix}0&I_n\\ I_n&0\end{pmatrix}.
\label{eq:symplectic-form}
\end{equation}
For $L\in\{A,B\}$, define the binary anticommutation matrix
\begin{equation}
(\Sigma_L)_{ij}:=(\ell_i^L)^TJ_{n_L}\ell_j^L\in\F_2.
\label{eq:Sigma-def}
\end{equation}
$\Sigma$ can be understood as the adjacency matrix of the anti-commutation graph of the Hamiltonian. $m$ Pauli terms in the Hamiltonian lead to an $m$-vertex graph, and two vertices are connected by an edge only when the two corresponding Paulis are anti-commuting.

We can view $A$ and $B$ as the parity check matrices of linear codes. Then, by defining the kernels of $A$ and $B$ as
\begin{equation}
K_A:=\ker A,\qquad K_B:=\ker B,
\label{eq:relation-codes-def}
\end{equation}
the distance of the code related to $A$ and $B$ is 
\begin{equation}
d_A:=\min_{0\ne c\in K_A}|c|,\qquad d_B:=\min_{0\ne c\in K_B}|c|.
\end{equation}
The mapping 
\begin{equation}
z\rightarrow Lz
\end{equation}
is injective on $\{z:|z|\le D\}$ whenever $2D<d_L$, for $L=A,B$.
An empty minimum is $+\infty$; for finite $d_L$ the largest such integer radius is $\floor{(d_L-1)/2}$.

When $K_A\subseteq K_B$, the induced quotient map has kernel $K_B/K_A$. Its nonzero classes are represented by the residual relations in $K_B\setminus K_A$. The relative distance is
\begin{equation}
d_{A\to B}=\min_{c\in K_B\setminus K_A}|c|.
\end{equation}
The mapping 
\begin{equation}
Az\rightarrow Bz
\end{equation}
is injective on $\mathcal S_D(A)$ exactly when $2D<d_{A\to B}$, as proved in Proposition~\ref{thm:injectivity-radius}. 

With $K_A$ and $K_B$, we can also define the ordered signed-relation products
\begin{equation}
R_A(s):=\prod_{i=1}^m(\alpha_iP_A(a_i))^{s_i}=\zeta_A(s)I,~s\in K_A.
\label{eq:idea-zeta}
\end{equation}
\begin{equation}
R_B(s):=\prod_{i=1}^m(\beta_iP_B(b_i))^{s_i}=\zeta_B(s)I,~s\in K_B.
\end{equation}

Products of at most $D$ indexed source terms contain every Pauli label that can occur in $Q_D(H_A)$. Since $P_n(g)P_n(h)=i^{\mu_n(g,h)}P_n(g+h)$, the resulting support envelope is $\mathcal S_D(A)$
\begin{equation}
\mathcal S_D(A):=\{Az: |z|\le D\}\subseteq \im A,
\end{equation}
and we have
\begin{equation}
Q_D(H_A)=\sum_{u\in\mathcal S_D(A)}C_A(u)P_{A}(u).
\end{equation}
Define $|P_n(g)\rangle:=(P_n(g)\otimes I)|\Phi_n\rangle$. These canonical Pauli Choi vectors are orthonormal. With
\begin{equation}
 N_L^2:=2^{-n_L}\Tr(Q_D(H_L)^2),\qquad L\in\{A,B\},
\end{equation}
we have
\begin{equation}
\ket{\Psi_A^Q}=\frac{1}{N_A}\sum_{u\in\mathcal S_D(A)}C_A(u)|P_{A}(u)\rangle.
\end{equation}

To enable the transduction, we want $Q_D(H_B)$ has the form
\begin{equation}\label{hbform}
Q_D(H_B)=\sum_{u\in\mathcal S_D(A)}\chi(u)C_A(u)P_{B}(M u),
\end{equation}
where 
\begin{equation}
M:\im A\to\im B,\qquad M(Az)=Bz.
\end{equation}
Thus, our goal, the state $\ket{\Psi_B^Q}$, has the form
\begin{equation}
\ket{\Psi_B^Q}=\frac{1}{N_A}\sum_{u\in\mathcal S_D(A)}\chi(u)C_A(u)|P_{B}(Mu)\rangle.
\end{equation}

Sufficient conditions for Equation~\eqref{hbform} are summarized in Theorem~\ref{thm:algebraic-main}. The proof combines Proposition~\ref{thm:injectivity-radius} with Proposition~\ref{prop:coefficient-matching}. In \cite{ShangFranca2026Transduction}, we also provide an algebraic homomorphism view of Theorem~\ref{thm:algebraic-main}.
\begin{theorem}[Algebraic transduction and residual distance]
\label{thm:algebraic-main}
Let $Q_D$ be a real polynomial of degree at most $D$.  Suppose:
\begin{enumerate}[label=(\alph*),leftmargin=2em]
    \item Same anti-commutation graph $\Sigma_A=\Sigma_B$;
    \item Kernel inclusion $K_A\subseteq K_B$;
    \item Signed-relation compatibility $\zeta_A(s)=\zeta_B(s)
        \qquad\text{for every }s\in K_A.$
    \item Polynomial degree $2D<d_{A\to B}$.
\end{enumerate}
Then, we have
\begin{equation}
Q_D(H_A)=\sum_{u\in\mathcal S_D(A)}C_A(u)P_{A}( u)\qquad\text{implies}\qquad Q_D(H_B)=\sum_{u\in\mathcal S_D(A)}\chi(u)C_A(u)P_{B}(M u).
\end{equation}
If $Q_D(H_A)\ne0$, then $N_A=N_B>0$, so this also yields the normalized state identity above.
\end{theorem}

\noindent\textit{\textbf{Efficient verification.---}}
The algebraic conditions (a)--(c) of
Theorem~\ref{thm:algebraic-main} require checking equality of the
anticommutation matrices, inclusion of the source relation kernel in
the target kernel, and agreement of the signed relation phases.
For explicitly supplied Pauli labels and signs, these checks can be
performed in polynomial time using binary linear algebra and Pauli
phase arithmetic. In particular, kernel inclusion and signed-phase
compatibility need only be checked on a basis of $K_A$, rather than
on every source relation. The degree condition (d), efficient relative
decoding, and efficient source-state preparation remain separate
requirements.

\begin{figure}[t]
\centering
\[
\Qcircuit @C=0.92em @R=1.1em {
\lstick{S:\ \ket{\Psi_A^Q}} & \gate{U_{\Bell,A}} & \multigate{1}{U_M} & \gate{U_\chi} & \qw & \targ & \qw & \qw & \qw & \rstick{\ket{0}} \qw \\
\lstick{T:\ \ket{0}} & \qw & \ghost{U_M} & \qw & \multigate{1}{U_{\rm dec}} & \qw & \multigate{1}{U_{\rm dec}^\dagger} & \gate{U_{\Bell,B}^\dagger} & \qw & \rstick{\ket{\Psi_B^Q}} \qw \\
\lstick{R:\ \ket{0}} & \qw & \qw & \qw & \ghost{U_{\rm dec}} & \ctrl{-2} & \ghost{U_{\rm dec}^\dagger} & \qw & \qw & \rstick{\ket{0}} \qw
}
\]

\caption{The transduction workflow.  The circuit starts from the source polynomial-filtered Choi state $\ket{\Psi_A^Q}$. The first gate $U_{\Bell,A}$ converts it into the source Bell-label state, where the relative-decoding operations $U_M$, $U_\chi$, and $U_{\rm dec}$ are performed.  After the source and recovery registers are uncomputed, the target label state is converted back to the target Choi state by $U_{\Bell,B}^\dagger$. }
\label{fig:transduction-circuit}
\end{figure}
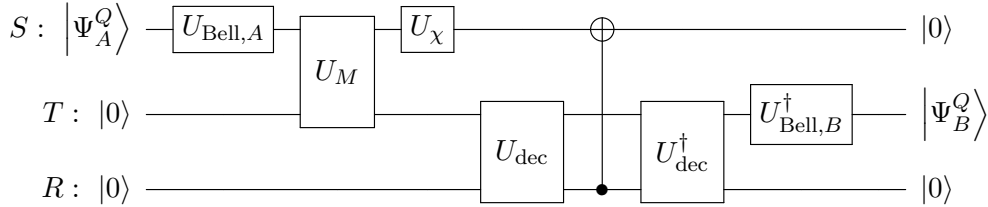

\noindent\textbf{Circuit.}
Given the source polynomial-filtered Choi state, the circuit in
Figure~\ref{fig:transduction-circuit} prepares the corresponding target
state. The source Bell transform acts as
\begin{equation}
U_{\Bell,A}|P_{A}(u)\rangle=|u\rangle,
\end{equation}
which is the tensor product of the following single-pair Bell transform circuit, followed by a fixed regrouping of label bits
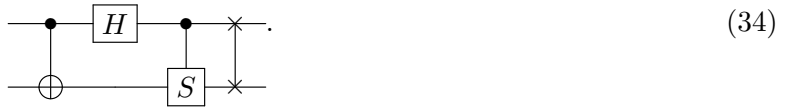
\begin{equation}
\Qcircuit @C=1.05em @R=0.9em {
\lstick{} & \ctrl{1} & \gate{H} & \ctrl{1} & \qswap & \rstick{} \qw \\
\lstick{} & \targ & \qw & \gate{S} & \qswap \qwx & \rstick{} \qw
}.
\end{equation}
Applying $U_{\Bell,A}$ gives
\begin{equation}
\ket{\widetilde\Psi_A^Q}_S=\frac{1}{N_A}\sum_{u\in\mathcal S_D(A)}C_A(u)\ket{u}_S,
\end{equation}
where we use $S$ to label the source register.

The unitary $U_M$ computes $Mu$ into the target register:
\begin{equation}
\ket{\widetilde\Psi_A^Q}_S |0\rangle_T\xrightarrow{\ U_M\ }\frac{1}{N_A}\sum_{u\in\mathcal S_D(A)}C_A(u)\ket{u}_S|Mu\rangle_T,
\end{equation}
where we use $T$ to label the target register.

The unitary $U_\chi$ applies the phase correction:
\begin{equation}
\frac{1}{N_A}\sum_{u\in\mathcal S_D(A)}C_A(u)\ket{u}_S|Mu\rangle_T\xrightarrow{\ U_\chi\ }\frac{1}{N_A}\sum_{u\in\mathcal S_D(A)}\chi(u)C_A(u)\ket{u}_S|Mu\rangle_T.
\end{equation}

Relative decoding erases the source label using a clean recovery register $R$. Apply $U_{\rm dec}$ to compute $u$ from $Mu$ into $R$, XOR $R$ into $S$, and apply $U_{\rm dec}^\dagger$ to reset $R$. Writing this composite operation as $U_{\rm erase}$ and suppressing the restored register $R$, the output is
\begin{equation}
\frac{1}{N_A}\sum_{u\in\mathcal S_D(A)}\chi(u)C_A(u)\ket{u}_S|Mu\rangle_T\xrightarrow{\ U_{\rm erase}\ }|0\rangle_S\frac{1}{N_A}\sum_{u\in\mathcal S_D(A)}\chi(u)C_A(u)|Mu\rangle_T.
\end{equation}

The target register is now in the Bell-label state
\begin{equation}
\ket{\widetilde\Psi_B^Q}_T=U_{\Bell,B}\ket{\Psi_B^Q}_T,
\end{equation}
Applying $U_{\Bell,B}^\dag$ yields $\ket{\Psi_B^Q}_T$.
Appendix~\ref{subsec:circuit-workflow} gives the full circuit analysis.

The circuit has the following efficiency guarantee.
\begin{theorem}[Efficient polynomial-filter transduction]
\label{thm:efficient-poly}
Assume the hypotheses of Theorem~\ref{thm:algebraic-main} and $Q_D(H_A)\ne0$. Suppose further that:
\begin{enumerate}[label=(\alph*),leftmargin=2em]
    \item the source polynomial-filtered Choi state $\ket{\Psi_A^Q}$ is efficiently preparable;
    \item there is an efficient coherent relative decoder
    \begin{equation}
        U_{\rm dec}:\ket{Mu}\ket{0}\longmapsto \ket{Mu}\ket{u},
        \qquad u\in\mathcal S_D(A).
    \end{equation}
\end{enumerate}
Then an efficient quantum circuit prepares $\ket{\Psi_B^Q}$. 
\end{theorem}
The phase correction $\chi$ is efficiently computable by
Subsection~\ref{subsec:chi-efficient}.

\section{DQI and HDQI as relation-free Hamiltonian transduction}
\label{sec:dqi-hdqi-special}

Ordinary DQI and the original HDQI construction are the relation-free
special cases of our framework.  In both cases, the source relation
kernel is trivial,
\begin{equation}
    K_A=\{0\},
\end{equation}
so the source label is the complete term-selection vector and relative
decoding reduces to ordinary decoding.  The appropriate relation-free
source, however, depends on the target anticommutation matrix.  A
product Hamiltonian suffices for diagonal DQI, whereas a general
noncommuting HDQI target requires a noncommuting relation-free source
with the same anticommutation graph. In \cite{ShangFranca2026Transduction}, we already provide a brief explanation for these results, and here is a more detailed explanation.

\paragraph{Diagonal DQI and Max-XOR-SAT.}
In this paragraph, diagonal Pauli labels are written using only their $Z$ coordinates; their zero $X$ rows are suppressed.
Let
\begin{equation}
    B=[b_1,\ldots,b_m]\in\mathbb F_2^{n\times m},
    \qquad
    v\in\mathbb F_2^m,
\end{equation}
and consider the parity constraints
\begin{equation}
    B^T x=v
    \pmod 2.
\end{equation}
The associated Max-XOR score and diagonal Hamiltonian are
\begin{equation}
\begin{aligned}
    g_{B,v}(x)
    &:=
    \sum_{i=1}^{m}
    (-1)^{v_i+b_i\cdot x},
    \label{eq:main-xor-score}
\end{aligned}
\end{equation}
\begin{equation}
\begin{aligned}
    H_{B,v}
    &:=
    \sum_{i=1}^{m}
    (-1)^{v_i}Z(b_i),
    \label{eq:main-xor-hamiltonian}
\end{aligned}
\end{equation}
where
\begin{equation}
    Z(b_i)
    :=
    \prod_{q:(b_i)_q=1}Z_q.
\end{equation}
Thus
\begin{equation}
    H_{B,v}\ket{x}
    =
    g_{B,v}(x)\ket{x}.
\end{equation}

For a real polynomial \(Q_D\) of degree at most \(D\), with \(Q_D(H_{B,v})\ne0\), the DQI output is the amplitude-encoded state below. Its preparation requires the pilot and a decoder on the degree-$D$ support (in the relation-free construction, $2D<d_B$).
\begin{equation}
    \ket{\mathrm{DQI}_{Q_D}}
    \propto
    \sum_{x\in\mathbb F_2^n}
    Q_D\!\left(g_{B,v}(x)\right)\ket{x},
    \label{eq:main-dqi-state}
\end{equation}
whereas diagonal HDQI retains a coherent reference copy:
\begin{equation}
\begin{aligned}
    \ket{\mathrm{HDQI}_{Q_D}}
    &=\frac{(Q_D(H_{B,v})\otimes I)\ket{\Phi_n}}
      {\norm{(Q_D(H_{B,v})\otimes I)\ket{\Phi_n}}}
      \\
    &\propto\sum_{x\in\mathbb F_2^n}
       Q_D\!\left(g_{B,v}(x)\right)\ket{x}\ket{x}.
    \label{eq:main-diagonal-hdqi-state}
\end{aligned}
\end{equation}
Hence DQI and diagonal HDQI differ only by the coherent copy isometry
\(\ket{x}\mapsto\ket{x}\ket{x}\).

The corresponding relation-free source is the product Hamiltonian
\begin{equation}
    H_{\mathrm{sel}}
    :=
    Z_1+\cdots+Z_m,
    \qquad
    A=I_m,
    \qquad
    K_A=\{0\}.
    \label{eq:main-selection-hamiltonian}
\end{equation}
The usual DQI pilot~\cite{JordanEtAlDQI} can be written as a weighted Dicke superposition
\begin{equation}
    \ket{R_{Q_D}^{\mathrm{DQI}}}
    \propto
    \sum_{k=0}^{D}
    \beta_k
    \sum_{\substack{z\in\mathbb F_2^m\\|z|=k}}
    \ket{z},
    \label{eq:main-dqi-pilot}
\end{equation}
where the coefficients \(\beta_k\) are determined by \(Q_D\).  Appendix~\ref{app:dqi-hdqi-relation-free} proves
\begin{equation}
    U_{\Bell,Z}
    (Q_D(H_{\mathrm{sel}})\otimes I)\ket{\Phi_m}
    \propto
    \ket{R_{Q_D}^{\mathrm{DQI}}}.
    \label{eq:main-dqi-pilot-bell-state}
\end{equation}
Equivalently, if \(V_{\mathrm{copy}}\ket{x}=\ket{x}\ket{x}\), then
\begin{equation}
    V_{\mathrm{copy}}H^{\otimes m}
    \ket{R_{Q_D}^{\mathrm{DQI}}}
    \propto
    (Q_D(H_{\mathrm{sel}})\otimes I)\ket{\Phi_m}.
    \label{eq:main-dqi-hadamard-copy}
\end{equation}
Indeed,
\begin{equation}
    H_{\mathrm{sel}}\ket{x}
    =
    (m-2|x|)\ket{x},
\end{equation}
so the copied post-Hadamard pilot is a polynomial-filtered Choi state
of a product Hamiltonian.

In the Bell-label basis, the source-to-target map is therefore
\begin{equation}
    \ket{z}
    \longmapsto
    (-1)^{v\cdot z}\ket{Bz}.
    \label{eq:main-dqi-transduction}
\end{equation}
Since \(K_A=0\), the decoder must recover the complete vector,
\begin{equation}
    Bz\longmapsto z,
\end{equation}
and
\begin{equation}
    K_B/K_A\cong K_B,
    \qquad
    d_{A\to B}=d_B.
    \label{eq:main-dqi-distance}
\end{equation}
Thus diagonal DQI is precisely the product-source, relation-free
endpoint of the transduction framework.

\paragraph{General noncommuting HDQI.}
The product source in Eq.~\eqref{eq:main-selection-hamiltonian} has
vanishing anticommutation matrix and therefore cannot be paired
directly with a genuinely noncommuting target.  Consider instead
\begin{equation}
    H_B
    =
    \sum_{i=1}^{m}
    \sigma_i P_B(b_i),
    \qquad
    \sigma_i\in\{\pm1\},
    \label{eq:main-general-hdqi-target}
\end{equation}
with label matrix
\begin{equation}
    B=[b_1,\ldots,b_m]
\end{equation}
and anticommutation matrix
\begin{equation}
    \Sigma
    :=
    B^T J_{n_B}B.
\end{equation}
Let \(T\in\mathbb F_2^{m\times m}\) be defined by
\begin{equation}
    T_{ji}
    :=
    \begin{cases}
        \Sigma_{ji}, & j<i,\\
        0,           & j\ge i,
    \end{cases}
    \qquad
    T+T^T=\Sigma.
    \label{eq:main-triangular-sigma}
\end{equation}
On an \(m\)-qubit source register, define
\begin{equation}
    \widehat P_i
    :=
    X_i
    \prod_{j<i}
    Z_j^{\Sigma_{ji}},
    \qquad
    \widehat H_G
    :=
    \sum_{i=1}^{m}\widehat P_i.
    \label{eq:main-universal-source}
\end{equation}
The binary label matrix of this source is
\begin{equation}
    \widehat A
    =
    \begin{pmatrix}
        I_m\\
        T
    \end{pmatrix}.
    \label{eq:main-universal-source-matrix}
\end{equation}
Consequently,
\begin{equation}
    \ker\widehat A=\{0\},
    \qquad
    \widehat A^T J_m\widehat A
    =
    T+T^T
    =
    \Sigma.
    \label{eq:main-universal-source-properties}
\end{equation}
Thus \(\widehat H_G\) is relation-free while reproducing exactly the
pairwise anticommutation pattern of the target. This relation-free Hamiltonian is also observed in \cite{HastingsSchmidhuber2026RedundancyFree} where they further show the hardness of this type of Hamiltonian and the cost of preparing its TFD.

To identify its polynomial state, write
\begin{equation}
    Q_D(x)
    =
    \sum_{k=0}^{D}q_kx^k,
\end{equation}
and let \(\beta_G^{(k)}(y)\) denote the graph-dependent coefficient
obtained by collecting all length-\(k\) words with parity vector
\(y\in\mathbb F_2^m\), including the signs generated by reordering
anticommuting terms.  Define
\begin{equation}
    a_G(y)
    :=
    \sum_{k=0}^{D}
    q_k\beta_G^{(k)}(y).
    \label{eq:main-beta-amplitude}
\end{equation}
The general HDQI pilot~\cite{HamiltonianDQI} is
\begin{equation}
    \ket{\mathcal R_D(G)}
    :=
    \frac{1}{N_G}
    \sum_{y\in\mathbb F_2^m}
    a_G(y)\ket{y}.
    \label{eq:main-general-hdqi-pilot}
\end{equation}
Because \(\widehat A\) is injective, the Choi vectors
\begin{equation}
    (\widehat P^{\,y}\otimes I)\ket{\Phi_m},
    \qquad
    \widehat P^{\,y}
    :=
    \prod_{i=1}^{m}\widehat P_i^{y_i},
\end{equation}
are mutually orthogonal.  A graph-adapted Bell-label transform
\(U_G^{\mathrm{lab}}\) can therefore be chosen so that
\begin{equation}
    U_G^{\mathrm{lab}}
    (\widehat P^{\,y}\otimes I)\ket{\Phi_m}
    =
    \ket{y}\ket{0^m}.
\end{equation}
Appendix~\ref{app:dqi-hdqi-relation-free} gives an efficient construction of this transform and proves, for $Q_D(\widehat H_G)\ne0$, the exact identity
\begin{equation}
    U_G^{\mathrm{lab}}
    \frac{
        (Q_D(\widehat H_G)\otimes I)\ket{\Phi_m}
    }{
        \left\|
        (Q_D(\widehat H_G)\otimes I)\ket{\Phi_m}
        \right\|
    }
    =
    \ket{\mathcal R_D(G)}\ket{0^m}.
    \label{eq:main-hdqi-pilot-source-state}
\end{equation}
Hence the noncommuting HDQI pilot is precisely the Bell-label
polynomial Choi state of the universal relation-free source
\(\widehat H_G\), up to a fixed Bell-basis gauge and a clean ancillary
register.

Since \(K_{\widehat A}=0\), the decoder again recovers the complete
selection vector,
\begin{equation}
    By\longmapsto y,
\end{equation}
and
\begin{equation}
    K_B/K_{\widehat A}\cong K_B,
    \qquad
    d_{\widehat A\to B}=d_B,
    \qquad
    2D<d_B.
    \label{eq:main-general-hdqi-distance}
\end{equation}
Thus the relative distance reduces to the ordinary distance.


\section{Thermofield-double transduction}
\label{sec:tfd-main}

\begin{figure}[t]
\centering
\begingroup
\definecolor{figTwoBlue}{HTML}{DAE8FC}
\definecolor{figTwoBlueEdge}{HTML}{6C8EBF}
\definecolor{figTwoGreen}{HTML}{D5E8D4}
\definecolor{figTwoGreenEdge}{HTML}{82B366}
\definecolor{figTwoGold}{HTML}{FFF2CC}
\definecolor{figTwoGoldEdge}{HTML}{D6B656}
\resizebox{0.94\linewidth}{!}{%
\begin{tikzpicture}[
  x=1cm,y=1cm,
  every node/.style={text=black,font=\sffamily\fontsize{10}{12}\selectfont},
  fig2 text/.style={align=center,inner sep=0pt,outer sep=0pt},
  fig2 state/.style={fig2 text,font=\fontsize{12}{14}\selectfont},
  fig2 heading/.style={fig2 text,font=\sffamily\fontsize{12}{14}\selectfont},
  fig2 arrow/.style={-{Stealth[length=2.2mm,width=1.6mm]},line width=.65pt},
  fig2 leader/.style={line width=.6pt}
]
  \path[use as bounding box] (0,0) rectangle (16.10,10.05);

  \draw[fill=figTwoBlue,draw=figTwoBlueEdge,line width=.7pt,rounded corners=9pt]
    (7.25,6.30) rectangle (15.90,9.45);
  \draw[fill=figTwoGreen,draw=figTwoGreenEdge,line width=.7pt,rounded corners=7pt]
    (7.25,.16) rectangle (15.90,2.23);
  \draw[fill=figTwoGold,draw=figTwoGoldEdge,line width=.7pt,rounded corners=9pt]
    (1.10,7.25) rectangle (3.85,9.45);

  \node[fig2 heading] at (11.575,9.85) {Source};
  \node[fig2 heading] at (11.575,2.64) {Target};

  \node[fig2 state,anchor=west] at (7.55,8.94)
    {$\ket{\TFD_A(\beta)}$};
  \draw[fig2 leader] (8.62,8.57) -- (8.62,8.23) -- (9.18,8.23);
  \node[fig2 text,anchor=west] at (9.32,8.23)
    {Purification of $\rho_{A,\beta}$};
  \node[fig2 text] at (11.575,7.65)
    {Within trace distance $\eta$ of $\lvert\Psi_A^Q\rangle$};
  \node[fig2 text,text width=7.8cm] at (11.575,6.90)
    {\shortstack{Prepared by another method, e.g.\\a purified Lindbladian Gibbs sampler}};

  \node[fig2 state] at (2.475,8.94)
    {$\lvert\widetilde{\TFD}_A\rangle$};
  \node[fig2 text,text width=2.30cm] at (2.475,8.16)
    {Bell-label state};
  \node[fig2 state] at (2.475,7.60)
    {$\approx\lvert\widetilde\Psi_A^Q\rangle$};
  \draw[fig2 arrow] (7.10,8.94) -- (3.86,8.94);
  \node[fig2 text] at (5.48,9.29) {Bell transform};
  \node[fig2 state] at (5.48,8.50) {$U_{\Bell,A}$};

  \draw[fig2 arrow] (1.09,8.94) -- (.30,8.94) -- (.30,1.53) -- (1.69,1.53);
  \node[fig2 state] at (2.63,1.53)
    {$\lvert\widetilde\Omega_D\rangle$};
  \draw[fig2 arrow] (3.87,1.53) -- (7.10,1.53);
  \node[fig2 text] at (5.48,1.88) {Bell transform};
  \node[fig2 state] at (5.48,1.07) {$U_{\Bell,B}^{\dagger}$};

  \node[fig2 text,anchor=west,font=\sffamily\fontsize{9}{11}\selectfont]
    at (1.10,6.77) {Polynomial comparison only ($|z|\le D$)};
  \node[fig2 state,anchor=west] at (1.11,6.19) {$\ket{Az}\ket0$};
  \draw[fig2 arrow] (1.74,5.86) -- (1.74,4.95);
  \node[fig2 text,anchor=west] at (2.06,5.40)
    {\shortstack{Matrix multiplication\\and phase correction}};
  \node[fig2 state,anchor=west] at (.58,4.49)
    {$\chi(Az)\ket{Az}\ket{Bz}$};
  \draw[fig2 arrow] (1.74,4.11) -- (1.74,3.15);
  \node[fig2 text,anchor=west] at (2.06,3.63)
    {\textcolor{red}{\shortstack{Relative decoding of\\$K_B/K_A$}}};
  \node[fig2 state,anchor=west] at (.74,2.72)
    {$\chi(Az)\ket0\ket{Bz}$};

  \node[fig2 text,text width=8.10cm] at (11.575,4.57)
    {Source-to-target Hamiltonian transduction\\
     requires $H_A$ and $H_B$ satisfy\\
     the algebraic compatibility};

  \node[fig2 state,anchor=west] at (7.55,1.53)
    {$\ket{\Omega_D}\approx\ket{\TFD_B(\beta)}$};
  \draw[fig2 leader] (8.62,1.15) -- (8.62,.82) -- (9.18,.82);
  \node[fig2 text,anchor=west] at (9.32,.82)
    {Approximate target TFD};
  \node[fig2 text] at (11.575,.38)
    {Joint trace-distance error $\le2\eta$};
\end{tikzpicture}%
}
\endgroup
\caption[TFD transduction via relative decoding]{TFD transduction via relative decoding. }
\label{fig2}
\end{figure}
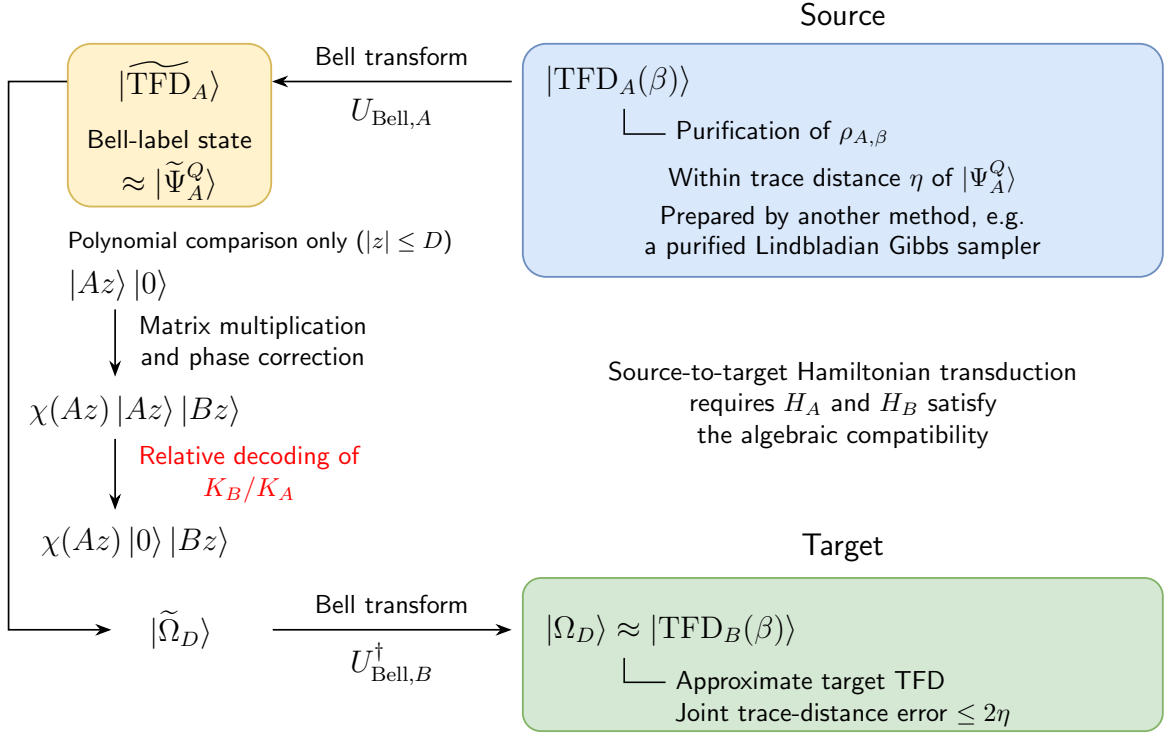
The thermal protocol uses the same relative-decoding circuit as polynomial
transduction but the input is the source TFD as shown in Fig.~\ref{fig2}.

\subsection{Available source states}
The source resource is a circuit preparing the canonical physical--reference
state. Coherent Gibbs-state constructions provide such circuits under their
gap and implementation assumptions~\cite{ChenEtAlNature2025,RouzeFrancaAlhambra}.
For bounded-strength, finite-range one-dimensional spin chains,
Corollary~I.1 of~\cite{BergamaschiChen1D} supplies
canonical-TFD preparation at every fixed finite inverse temperature.
Its cost depends on temperature and the requested precision.

\subsection{Canonical input and available decoding degree}
For $L\in\{A,B\}$ and $\beta\ge0$, define
\begin{equation}
 Z_L(\beta):=\Tr(e^{-\beta H_L}),\qquad
 \ket{\TFD_L(\beta)}:=
 \frac{(e^{-\beta H_L/2}\otimes I)\ket{\Phi_{n_L}}}
 {\sqrt{2^{-n_L}Z_L(\beta)}}.
 \label{eq:tfd-main-definition}
\end{equation}
The physical marginal is $e^{-\beta H_L}/Z_L(\beta)$. The algorithm takes
an efficiently prepared approximation to this \emph{canonical purification}
on the specified physical and reference registers.

Assume the algebraic compatibility conditions of
Theorem~\ref{thm:algebraic-main}. Let $D_{\rm dec}$ be a degree for which
an efficient coherent relative decoder is supplied on
$\mathcal S_{D_{\rm dec}}(A)$. When $d_{A\to B}<\infty$, require
\begin{equation}
 0\le D_{\rm dec}\le
 \left\lfloor\frac{d_{A\to B}-1}{2}\right\rfloor.
 \label{eq:tfd-decoding-degree}
\end{equation}

\subsection{Error, degree, and temperature guarantees}
\label{subsec:tfd-error-degree-temperature}

We combine the direct transduction error estimate with the Taylor--Poisson
approximation certificate of \cite{ShangFranca2026Transduction}.
An alternative approximation via the Chebyshev polynomials is given in HDQI~\cite{HamiltonianDQI}

\begin{theorem}[TFD transduction and inverse-temperature window]
\label{thm:tfd-transduction}
Let $\Lambda\ge\|H_A\|$, $\beta\ge 0$, $0<\eta<1$, and
$\ell:=\ln(1/\eta)$. Let $D_{\mathrm{dec}}$ be the efficient relative decoding degree ($2D_{\mathrm{dec}}< d_{A\rightarrow B}$). Suppose the canonical source TFD is efficiently prepared to
trace-distance error $\delta_{\mathrm{src}}$.
If an integer $0\le D\le D_{\mathrm{dec}}$ satisfies
\begin{equation}
\label{eq:tfd-degree-bound}
    D+1
    \ge
    \beta\Lambda
    +\sqrt{2\beta\Lambda\,\ell}
    +\frac{\ell}{3},
\end{equation}
then the transduction circuit prepares the target TFD to trace-distance error at most
\begin{equation}
\label{eq:tfd-error-bound}
    \delta_{\mathrm{src}}+2\eta+\delta_{\mathrm{circ}},
\end{equation}
where $\delta_{\mathrm{circ}}$ is the output implementation
error. The same bound holds for the target Gibbs marginal.

Writing $k:=D_{\mathrm{dec}}+1$, for $\Lambda>0$ and
$k\ge\ell/3$, the degree condition is feasible precisely when
\begin{equation}
\label{eq:tfd-finite-temperature-window}
    0\le\beta\le
    \frac{1}{\Lambda}
    \left(
        \sqrt{k+\frac{\ell}{6}}
        -
        \sqrt{\frac{\ell}{2}}
    \right)^2.
\end{equation}
For total error $0<\epsilon<1$, take $\eta=\frac{\epsilon-\delta_{\mathrm{src}}
              -\delta_{\mathrm{circ}}}{2}>0$ and $\ell=\ln(1/\eta)$. If $D_{\mathrm{dec}}\ge\delta_{\mathrm{dec}}m$,
$\Lambda\le Jm$, and $\ell=o(m)$ for fixed
$\delta_{\mathrm{dec}},J>0$, then every fixed
\begin{equation}
\label{eq:tfd-temperature-window}
    0\le\beta<\frac{\delta_{\mathrm{dec}}}{J}
\end{equation}
is certified for sufficiently large $m$.

\end{theorem}
Thus an extensive spectral bound and efficient decoding at linear radius
certify a constant-inverse-temperature preparation window.

The potential advantage concerns target sampling. Within the certified
low-degree regime, target observables can already be estimated using the
source Gibbs state through the algebraic
homomorphism~\cite{ShangFranca2026Transduction}.

\section{Example: relation-free commuting source}
\label{sec:examples1}\label{sec:examples-main}

For commuting Hamiltonians, the original DQI and HDQI both only work for those with uniform coefficients, i.e. the target $H_B$ in Equation~\eqref{eq:hamiltonians} has $\lambda_i=1$ since for non-uniform coefficients, the initial state to be prepared can not be cast as a linear combination of Dicke states. Here, using the Hamiltonian transduction picture, we show that for any commuting target Hamiltonian $H_B$, the $D$-degree polynomial-filtered Choi state of a relation-free source $H_A$ is Clifford-equivalent to an MPS of bond
dimension at most $D+1$. A similar result is also observed in \cite{GeneralPauliHDQI} while our transduction picture provides a simpler proof. The following theorem summarize the result, and Appendix~\ref{sec: classicalsimu} proves this
bound and provides explicit tensors. Sequential MPS preparation
is efficient on quantum computers~\cite{SchonEtAl2005,MalzEtAl2024MPS}.
\begin{theorem}[Commuting relation-free sources]
\label{thm:commuting-relation-free-source}
Let
\begin{equation}
    H_A=\sum_{i=1}^{m}\lambda_i\alpha_iP_A(a_i),
    \qquad
    \lambda_i>0,\quad \alpha_i\in\{\pm1\},
\end{equation}
be an $n_A$-qubit Hamiltonian with pairwise commuting
Pauli terms and $K_A=\ker A=\{0\}$.
For every real polynomial $Q_D$ of degree at most $D$
with $Q_D(H_A)\ne0$, there is an efficiently computable
Clifford unitary $U$ such that
\begin{equation}
    (U\otimes U^*)\ket{\Psi_A^Q}
       =\ket{\phi_D}\otimes\ket{\Phi_{n_A-m}},
\end{equation}
where $\ket{\phi_D}$ has an exact open-boundary MPS
representation of bond dimension at most $D+1$ in
the paired physical--reference ordering.
Its tensors can be computed efficiently and consequently, $\ket{\Psi_A^Q}$ is efficiently preparable.
\end{theorem}

\section{Example: 1d source}
\label{sec:examples2}

We construct geometrically nonlocal target Hamiltonians whose TFDs can be
prepared from a one-dimensional source. The source preparation follows
from the spin-chain result of~\cite{BergamaschiChen1D}, and a classical
syndrome decoder supplies the relative decoder. The following theorem
states the resulting guarantee; the rest of the section gives the
construction and proof.
\begin{theorem}[Transduction from a one-dimensional source]
\label{thm:1d-source-transduction}
Consider $H_A,H_B$ in
\eqref{eq:example-1d-source-hamiltonian} and
\eqref{eq:example-code-target-hamiltonian}. Suppose $F$ has distance $d_F\ge \delta N$
and an efficient syndrome decoder correcting every
error of weight at most $\lfloor\tau N\rfloor$,
where $0<2\tau<\delta$ are fixed constants.
Then the pair is algebraically compatible and, with
$m=5N-1$,
\begin{equation}
    d_{B}=3,\qquad
    d_{A\rightarrow B}=d_F=\Theta(m),\qquad
    D_{\rm dec}=\lfloor\tau N\rfloor.
\end{equation}
For every fixed inverse temperature
\begin{equation}
    0\le\beta<\frac{\tau}{5 J_0},
\end{equation}
transduction from the efficiently preparable
one-dimensional source TFD prepares the target TFD
and its Gibbs marginal to fixed or inverse-polynomial
trace-distance error in polynomial time, for all
sufficiently large $N$.

If $F$ has bounded row and column weights, the target
has bounded Pauli locality and participation.
If its column supports have unbounded diameter in
the specified target layout, the target is not
finite-range one-dimensional in that layout.
\end{theorem}

\subsection{A one-dimensional source and its non-geometrically-local target}
\label{subsec:1d-fast-mixing-code-quotient}

Consider an open chain of $N\ge2$ two-qubit cells
$(q_i,c_i)$ and we call qubits with $c_i$ the central qubits. Define
\begin{equation}
\begin{aligned}
    H_A
    ={}&
    \sum_{i=1}^{N}
    \left(
        \lambda_i^xX_{q_i}
        +\lambda_i^zZ_{q_i}
        +\lambda_i^cZ_{c_i}
        +\lambda_i^{cq}Z_{c_i}Z_{q_i}
    \right)
    \\
    &+
    \sum_{i=1}^{N-1}
        \lambda_i^{zz}Z_{q_i}Z_{q_{i+1}},
    \label{eq:example-1d-source-hamiltonian}
\end{aligned}
\end{equation}
where all displayed coefficients are real and nonzero, with
$0<|\lambda_i^\mu|\le J_0$ for each displayed term and a
system-size-independent constant $J_0$.
Every term acts within one cell or two adjacent cells.
Thus $H_A$ is explicitly a nearest-neighbor chain of
four-dimensional sites, with
\begin{equation}
    n_A=2N,\qquad m=5N-1.
\end{equation}

Choose a full-row-rank binary parity-check matrix
\begin{equation}
    F\in\F_2^{(N-k)\times N},
    \qquad 0<k<N,
    \qquad
    d_F:=\min_{\substack{e\ne0\\Fe=0}}|e|.
    \label{eq:ex-F}
\end{equation}
The $N$ columns correspond to the central source spins.
Introduce $N-k$ target spins $R_1,\ldots,R_{N-k}$ and set
\begin{equation}
    Z_R(y):=\prod_{a=1}^{N-k}Z_{R_a}^{y_a}.
\end{equation}
Writing $f_i$ for column $i$ of $F$, replace
$Z_{c_i}$ by $Z_R(f_i)$ to obtain
\begin{equation}
\begin{aligned}
    H_B
    ={}&
    \sum_{i=1}^{N}
    \left(
        \lambda_i^xX_{q_i}
        +\lambda_i^zZ_{q_i}
        +\lambda_i^cZ_R(f_i)
        +\lambda_i^{cq}Z_R(f_i)Z_{q_i}
    \right)
    \\
    &+
    \sum_{i=1}^{N-1}
        \lambda_i^{zz}Z_{q_i}Z_{q_{i+1}}.
    \label{eq:example-code-target-hamiltonian}
\end{aligned}
\end{equation}
The target retains the $N$ quantum spins but replaces
the $N$ central spins by $N-k$ syndrome spins, giving
$n_B=2N-k$. The replacement acts only on commuting $Z$ factors and preserves all
quantum factors and signed coefficients. It therefore preserves ordered
products and signed scalar relations, proving algebraic conditions
(a)--(c) of Theorem~\ref{thm:algebraic-main}; the degree condition is checked below.

\paragraph{Source preparation.}
$H_A$ is a 1$D$ source. The on-site and nearest-neighbor operators in
\eqref{eq:example-1d-source-hamiltonian} have norms
at most $4J_0$ and $J_0$, respectively. Its canonical TFD is therefore efficiently preparable
at every fixed finite inverse temperature.

\paragraph{Target geometry.}
The source field $Z_{c_i}$ becomes
\begin{equation}
    Z_R(f_i)=\prod_{a:(f_i)_a=1}Z_{R_a},
\end{equation}
whose support is determined by $F$, rather than
adjacency along the source chain. Choose $F$ so that their supports have unbounded diameter
in the specified target layout.
Then the terms $Z_R(f_i)$ are geometrically nonlocal,
and $H_B$ is not a finite-range one-dimensional
Hamiltonian in that layout.

This does not require large Pauli weight.
If $F$ has column weight at most $d_v$, every target
term has weight at most $\max\{2,d_v+1\}$.
If its row weight is at most $d_c$, each $R$ spin
participates in at most $2d_c$ terms.
Thus sparse choices of $F$ can give bounded-locality
targets with geometrically nonlocal interactions.

\subsection{Relative decoding reduces to decoding
\texorpdfstring{$F$}{F}}
\label{subsec:examples-relative-decoding}

Let $\Gamma\in\F_2^{N\times(N-1)}$ be $\Gamma_{ij}=\delta_{ij}+\delta_{i,j+1}$.
Order the indexed terms as
$X_q,Z_q,Z_qZ_q,Z_c,Z_cZ_q$.
Their Pauli-label matrices are
\begin{equation}
    A=
    \begin{pmatrix}
        I_N&0&0&0&0\\
        0&0&0&0&0\\
        0&I_N&\Gamma&0&I_N\\
        0&0&0&I_N&I_N
    \end{pmatrix},
    \qquad
    B=
    \begin{pmatrix}
        I_N&0&0&0&0\\
        0&0&0&0&0\\
        0&I_N&\Gamma&0&I_N\\
        0&0&0&F&F
    \end{pmatrix}.
    \label{eq:ex-AB}
\end{equation}
The row blocks are $(x_q,x_c,z_q,z_c)$ for $A$
and $(x_q,x_R,z_q,z_R)$ for $B$.

For a selection vector
$s=(s_X,s_Z,s_{ZZ},s_c,s_{cq})$, set
\begin{equation}
    \xi:=s_Z+\Gamma s_{ZZ}+s_{cq},
    \qquad
    e:=s_c+s_{cq}.
\end{equation}
Then
\begin{equation}
    As=(s_X,0_N,\xi,e),
    \qquad
    Bs=(s_X,0_{N-k},\xi,Fe).
    \label{eq:ex-label-maps}
\end{equation}
Relative decoding therefore only needs
to recover the central label $e$ from its
syndrome $Fe$.

Suppose an efficient decoder for $F$
satisfies
\begin{equation}
    \operatorname{Dec}_F(Fe)=e
    \quad\text{whenever }|e|\le D_{\rm cl},
    \qquad
    2D_{\rm cl}<d_F.
\end{equation}
Then the relative decoder can be built from $\operatorname{Dec}_F$
\begin{equation}
    \operatorname{Dec}_{A\to B}(x,0_{N-k},z,y)
       :=(x,0_N,z,\operatorname{Dec}_F(y)).
    \label{eq:ex-decoder}
\end{equation}

\paragraph{Distance comparison.}
A residual relation $s\in K_B\setminus K_A$ has a nonzero central label
$e$ satisfying $Fe=0$.
Its selection weight is at least $|e|$.
Conversely, for every nonzero $e$ with $Fe=0$, the selection
$s=(0,0,0,e,0)$ is a residual relation of weight $|e|$.
Therefore
\begin{equation}
    d_{A\to B}
       =\min_{\substack{e\ne0\\Fe=0}}|e|
       =d_F.
    \label{eq:ex-relative-distance}
\end{equation}
Write $\alpha_i^\mu:=\operatorname{sgn}(\lambda_i^\mu)$ and
$\omega_i:=\alpha_i^z\alpha_i^c\alpha_i^{cq}\in\{\pm1\}$.
Then the source has the weight-three signed Pauli relations
\begin{equation}
    \alpha_i^zZ_{q_i}\,\alpha_i^cZ_{c_i}\,
    \bigl(\alpha_i^{cq}Z_{c_i}Z_{q_i}\bigr)=\omega_i I,
    \qquad i\in[N],
\end{equation}
Thus each relation involves
three indexed terms, giving $d_B\le3$.
The source labels are nonzero and distinct,
excluding source relations of weight one or two.
Since every new target relation has weight at
least $d_F$, whenever $d_F\ge4$ we obtain the ordinary distance-relative distance separation.

\subsection{Linear distance and constant inverse temperature}
\label{subsec:examples-temperature}

Choose a sparse parity-check family with
\begin{equation}
    d_F\ge\delta N,
    \qquad
    D_{\rm cl}=\floor{\tau N},
    \qquad
    0<2\tau<\delta,
    \label{eq:ex-code-family}
\end{equation}
for fixed $\delta,\tau>0$, together with an efficient
decoder through $D_{\rm cl}$.
Since $m=5N-1$, the construction gives, for sufficiently
large $N$,
\begin{equation}
    d_B=3,\qquad
    d_{A\rightarrow B}=\Theta(m),\qquad
    D_{\rm dec}=D_{\rm cl}=\Theta(m).
\end{equation}
The coefficient bound supplies the common norm bound $ \Lambda=J_0m=J_0(5N-1)$.

For fixed $\beta\ge0$ and $0<\eta<1$ with
$\ell:=\ln(1/\eta)=o(N)$, the Taylor--Poisson certificate
gives the sufficient degree
\begin{equation}
\begin{aligned}
    D_{\rm req}
    &:=
    \max\left\{
        0,\,
        \ceil{
            \beta\Lambda
            +\sqrt{2\beta\Lambda\,\ell}
            +\frac{\ell}{3}
        }-1
    \right\} \\
    &=5\beta J_0N+o(N).
\end{aligned}
\end{equation}
Comparing with $D_{\rm dec}=\tau N+O(1)$ shows that
$D_{\rm req}\le D_{\rm dec}$ for sufficiently large $N$
whenever
\begin{equation}
    0\le\beta<\frac{\tau}{5J_0}.
    \label{eq:ex-temperature}
\end{equation}
Under the same source-preparation and circuit-implementation
assumptions, Theorem~\ref{thm:tfd-transduction} therefore
certifies target-TFD preparation with total trace-distance
error at most $\delta_{\rm src}+2\eta+\delta_{\rm circ}$.

Regarding the candidates of the code family, one example is the binary expander
codes of Sipser and Spielman
\cite{SipserSpielman1996}. They have positive rate, linear minimum distance, bounded
parity-check row and column weights, and an efficient decoder
correcting every error pattern of weight at most
$\lfloor\tau_0N\rfloor$ for some constant $\tau_0>0$.

\section{Source design and structural constraints}
\label{sec:properties}

Throughout this section, $\Sigma_A=\Sigma_B=: \Sigma$, $K_A\subseteq K_B$, and $K_\Sigma:=\ker\Sigma$.
Proofs are in Appendices~\ref{app:graph-kernel-independence},
\ref{app:relative}, and~\ref{app:dense}.

\begin{property}[Existence of Hamiltonian]
Given a symmetric matrix $\Sigma\in\mathbb{F}_2^{m\times m}$ with zero diagonal and let $K\subseteq\mathbb{F}_2^m$ be a linear subspace. Then there exist an $n$-qubit Hamiltonian $H_L$ with the Pauli-label matrix $L\in\mathbb{F}_2^{2n\times m}$ such that
\begin{equation}
    \ker L=K,
    \qquad
    L^T J_n L=\Sigma,
\end{equation}
if and only if
\begin{equation}
    K\subseteq K_\Sigma.
\end{equation}
Moreover, a concrete construction of $H_L$ can be built on $n=q+s$ qubits with $\operatorname{rank}\Sigma=2q$ and $s=\dim\ker\Sigma-\dim K$. And we have $n\leq m$ and 
\begin{equation}
    n=m
    \quad\Longleftrightarrow\quad
    \Sigma=0
    \text{ and }
    K=\{0\}.
\end{equation}
\end{property}
\noindent Proofs can be found in Proposition~\ref{thm:prescribed-target-kernel}.

For a fixed anticommutation matrix, compatible source and target relation
kernels are exactly the chains $K_A\subseteq K_B\subseteq K_\Sigma$.
The relation-free realization in
Proposition~\ref{prop:universal-pauli-graph} supplies the DQI and HDQI source
states described in Section~\ref{sec:dqi-hdqi-special}.

\begin{property}[Kernel and relative distance]
Given two source Hamiltonians $H_A$ and $H_{A'}$ with the kernels $K_A$ and $K_{A'}$, the relative distance is monotone under absorbing source relations
\begin{equation}
K_{A'}\subseteq K_A\subseteq K_B \quad\Longrightarrow\quad d_{A'\to B}\le d_{A\to B}.
\end{equation}
When $K_{A'}=0$, $d_{A'\to B}=d_B$ and the relative distance $d_{A\to B}$ is strictly larger than the oridinary distance $d_B$ if and only if
\begin{equation}
\{c\in K_B: |c|=d_B\}\subseteq K_A.
\end{equation}
\end{property}
\noindent Proofs can be found in Subsection~\ref{sec: kernel affect distance}.

To increase the relative distance, the source should absorb the low-weight
target relations. The resulting degree window can be used when source
preparation and relative decoding remain efficient.

\begin{property}[Achievable distance]
\begin{enumerate}[label=(\roman*)]
\item Fixed source. Fixing $K_A$ and $K_\Sigma$ with $K_A\subsetneq K_\Sigma$, then the largest relative distance by varying the target is 
\begin{equation}
\max_{K_A\subsetneq K_B\subseteq K_\Sigma}d_{A\to B}=\max_{c\in K_\Sigma}\min_{s\in K_A}|c+s|.
\end{equation}
A one-dimensional extension
\begin{equation}
        K_B=K_A\oplus\Span\{c_\star\}
        \label{eq:main-one-dimensional-best-extension}
\end{equation}
attains the optimum whenever $c_\star+K_A$ is a farthest coset.

\item Fixed target. Fix $K_B$ and $K_\Sigma$ with $K_B\subseteq K_\Sigma$. Define $K_B^{<t}:=\Span\{c\in K_B:|c|<t\}$.
There exists a proper source kernel $K_A\subsetneq K_B$ making $d_{A\to B}\ge t$ if and only if
\begin{equation}
K_B^{<t}\subsetneq K_B.
\end{equation}
Any $K_A$ satisfying $K_B^{<t}\subseteq K_A\subsetneq K_B$ has the relative distance at least $t$.
\end{enumerate}
\end{property}
\noindent Proofs can be found in Subsection~\ref{subsec:linear-relative-distance-existence} and Subsection~\ref{subsec:fixed-target-source-design}.

For a fixed source, linear relative distance is possible precisely when
$K_\Sigma/K_A$ contains a coset at linear distance from $K_A$. For a fixed
target and a proper source kernel, it is possible when the target relations
below the desired linear threshold span a proper subspace of $K_B$.

\begin{property}[Finite distance implies relation-changing transductions]
Let
\begin{equation}
        \widetilde P_i^A:=\alpha_iP_A(a_i),
        \qquad
        \widetilde P_i^B:=\beta_iP_B(b_i).
\end{equation}
If a unitary satisfies
\begin{equation}
        \widetilde P_i^B
        =U\widetilde P_i^AU^\dagger
        \quad\text{for every }i,
\end{equation}
then
\begin{equation}
        K_A=K_B,
        \qquad
        d_{A\to B}=+\infty,
\end{equation}
and every signed relation phase is preserved.  If $U$ is Clifford, there is an invertible symplectic matrix $S$ such that
\begin{equation}
        B=SA,
        \qquad
        M(Az)=S(Az).
\end{equation}
Conversely,
\begin{equation}
        d_{A\to B}<\infty
        \quad\Longrightarrow\quad
        K_A\subsetneq K_B
        \quad\Longrightarrow\quad
        \text{no termwise unitary conjugation exists}.
\end{equation}
Every finite residual distance belongs to $\{1,\ldots,m\}$.
\end{property}
\noindent Proofs can be found in Subsection~\ref{sec:mapping-hierarchy}.

Termwise unitary conjugation preserves the relation kernels and transports
every polynomial degree. Finite relative distance certifies a change in the
indexed relation structure that no termwise unitary can realize. Relative
decoding implements this change on the degree-$D$ subspace whenever
$2D<d_{A\to B}$.

\begin{property}[Minimum-weight ordinary decoder solves the relative decoding]
Assume $K_A\subsetneq K_B$. Suppose an ordinary target decoder
    \begin{equation}
        \operatorname{Dec}_B:\operatorname{im}B
        \longrightarrow \mathbb{F}_2^m
    \end{equation}
    is a minimum-weight decoder
    \begin{equation}
        \operatorname{Dec}_B(y)=z^\star(y)
 \end{equation}
with
\begin{equation}
  z^\star(y)\in
        \arg\min_{z':\,Bz'=y}|z'|
 \end{equation}
Then
    \begin{equation}
\operatorname{Dec}_{A\to B}^{\mathrm{MW}}(y)
         :=
        A\operatorname{Dec}_B(y)=
        Az^\star(y)
    \end{equation}
is a valid relative decoder whenever
    \begin{equation}
        2D<d_{A\to B}.
    \end{equation}
\end{property}
\noindent Proofs can be found in Subsection~\ref{subsec:standard-vs-relative}.

This gives a relative decoder by composing minimum-weight decoding with
$A$. An efficient implementation may instead compute $Az$ directly from
$Bz$, as in the chain construction of Section~\ref{sec:examples2}.

\begin{property}[Kernel and interacting Hamiltonian\label{pro:inter}]
For a Pauli label $\ell_i=(x_i,z_i)$, define its support and Pauli weight by
\begin{equation}
        S_i:=\{q:(x_{i,q},z_{i,q})\ne(0,0)\},
        \qquad
        \operatorname{wt}_P(\ell_i):=|S_i|.
\end{equation}
The number of inter-qubit interaction terms is defined as
\begin{equation}
        t_{\rm int}(L)
        :=\#\{i:\operatorname{wt}_P(\ell_i)\ge2\}.
        \label{eq:interaction-count-properties}
\end{equation}

Let $H_L$ be a Hamiltonian whose unsigned Pauli labels are nonzero and
pairwise distinct, with label matrix $L$, and set $k:=\dim K_L$.
Let $L_1$ be the submatrix containing the weight-one columns of $L$, and
define $k_1:=\dim\ker L_1$. Then
\begin{equation}
t_{\rm int}(L)\ge k-k_1 \ge
        \max\left\{0,\dim K_L-\frac12\rank\Sigma_L\right\}.
\end{equation}
Consequently, if $H_L$ is a commuting Hamiltonian, then
\begin{equation}
t_{\rm int}(L)\ge k,
\end{equation}
and, after a Clifford change of basis, the family can be written as $m-k$ independent one-qubit $Z$ fields plus exactly $k$ additional commuting $Z$-string interactions. 
\end{property}
\noindent Proofs can be found in Subsection~\ref{sec: howmanyinter} and Subsection~\ref{sec:noncommuting}. 

For reduced commuting families, the relation-code dimension bounds the
number of interaction terms. Noncommuting on-site $X,Y,Z$ triples can
account for some relations without interactions, giving the correction
involving $\rank\Sigma_L$. Source design balances absorbing short target
relations against retaining an efficiently preparable source TFD.

\section{Discussion and outlook}
\label{sec:discussion}

We have given an explicit Gibbs-state preparation algorithm that combines
a noncommuting one-dimensional source with relative decoding. The source
contains the short target relations, while a classical syndrome decoder
resolves the remaining quotient code. This yields linear decoding degree
and a constant-inverse-temperature preparation window for the target family.

Further constructions can use sources with different relation structures
and efficient relative decoders to reach other temperature and accuracy
regimes. The homomorphic formulation of~\cite{ShangFranca2026Transduction}
also supports fermionic and qudit systems, as well as bosonic models with
controlled finite-dimensional truncations. Thermal applications in these
settings call for explicit source TFDs and relative decoders.

Two goals are to establish classical sampling hardness for target Gibbs
states at the certified temperatures and to compare transduction with
direct target-preparation methods. A family with polynomial-time
transduction and exponentially slow mixing for a specified target sampler
at the same temperature would demonstrate a separation between preparation
and equilibration by that dynamics.

\section{Acknowledgments and AI use statement}
\noindent\textbf{Acknowledgments. }The authors would like to thank Samuel Slezak for relevant discussions. Z.S. and D.S.F. acknowledge support by the ERC grant GIFNEQ 101163938. D.S.F. acknowledges financial support from the Novo Nordisk Foundation (Grant No. NNF20OC0059939 Quantum for Life). Y.W. acknowledges support from the Danish e-Infrastructure Consortium (DeiC) through a PhD scholarship in quantum algorithms and quantum software, grant no. 5259-00016B.

\noindent\textbf{AI use statement.} The core ideas of interpreting the pilot state preparations in DQI as the polynomial Choi state of the product source Hamiltonian, exploring non-trivial non-relation-free sources, and involving existing TFD preparations in the source stage are from human authors. ChatGPT 5.6 Sol, ChatGPT 6 Astra were used for formulations and for helping the authors realize that relative decoding is what is needed at the uncomputation stage. The same models were also used for proofreading.

\bibliographystyle{unsrt}
\bibliography{finalgibbs_references}

\clearpage
\appendix
\setcounter{section}{0}

\section{Pauli phase calculus and coefficient matching}
\label{app:phase}\label{sec:phase-compatibility}

This appendix supplies the Pauli algebra and coefficient comparison used in Theorem~\ref{thm:algebraic-main}.  The goal is to show that the phase relating source and target words is a function of the compressed source label $u=Az$, not of the particular term-selection vector $z$.

\subsection{Pauli phase arithmetic}
\label{subsec:pauli-phase-arithmetic}

For the Hermitian representative in Eq.~\eqref{eq:pauli-representative}, every $P_n(g)$ is Hermitian and satisfies $P_n(g)^2=I$.  We use this convention throughout because it keeps every Hamiltonian term Hermitian even when the Pauli label contains both $X$ and $Z$ on the same qubit.  The price is that products carry explicit fourth-root phases.

Let $g=(x,z)$ and $h=(x',z')$.  Starting from the definition,
\begin{equation}
\begin{aligned}
        P_n(g)P_n(h)
        &= i^{x\cdot z+x'\cdot z'}
           X^x Z^z X^{x'}Z^{z'}  \\
        &= i^{x\cdot z+x'\cdot z'}(-1)^{z\cdot x'}
           X^{x+x'}Z^{z+z'}       \\
        &= i^{x\cdot z+x'\cdot z'+2z\cdot x'-(x+x')\cdot(z+z')}
           P_n(g+h),
\end{aligned}
\end{equation}
where binary additions in the Pauli label are over $\F_2$, while the dot products in the exponent are evaluated as ordinary integers before reducing modulo $4$.  Thus
\begin{equation}
        P_n(g)P_n(h)=i^{\mu_n(g,h)}P_n(g+h),
        \label{eq:pauli-product-rule}
\end{equation}
with
\begin{equation}
        \mu_n(g,h)=x\cdot z+x'\cdot z'+2z\cdot x'-(x+x')\cdot(z+z')
        \pmod 4.
        \label{eq:mu-def}
\end{equation}
Exchanging $g$ and $h$ cancels the Hermitian-convention terms and leaves the usual symplectic commutation sign:
\begin{equation}
        P_n(g)P_n(h)=(-1)^{g^TJ_nh}P_n(h)P_n(g).
        \label{eq:commutation-rule}
\end{equation}
Consequently, $(\Sigma_L)_{ij}$ defined in Eq.~\eqref{eq:Sigma-def} records exactly whether the $i$th and $j$th generators commute or anticommute.  The condition $\Sigma_A=\Sigma_B$ means that source and target words acquire identical signs whenever their factors are reordered.  This is the only place where pairwise anticommutation enters the coefficient comparison; the remaining phase difference is carried by the canonical phases $\nu_A,\nu_B$ and the sign vector $r$ below.

\subsection{Words, parity vectors, and canonical products}
For comparisons between source and target we assume
$\Sigma_A=\Sigma_B=\Sigma$. For an individual family, the same
identities hold with its own anticommutation matrix $\Sigma_L$.

A word of length $\ell$ is an ordered tuple
\begin{equation}
        w=(i_1,\ldots,i_\ell),\qquad i_t\in[m].
\end{equation}
Its parity vector $z(w)\in\F_2^m$ is
\begin{equation}
        z(w)_i=\#\{t:i_t=i\}\pmod 2.
\end{equation}
Only this parity vector determines the final binary Pauli label, because every generator squares to identity.

For $z\in\F_2^m$, define the canonical ordered generator product
\begin{equation}
        \mathcal P_L(z):=\prod_{i=1}^m P_L(\ell_i^L)^{z_i},
        \label{eq:canonical-product}
\end{equation}
where the factors are always ordered by increasing $i$.  Its binary Pauli label is $Lz$.  Hence there is a unique phase $\nu_L(z)\in\{\pm 1,\pm i\}$ such that
\begin{equation}
        \mathcal P_L(z)=\nu_L(z)P_L(Lz).
        \label{eq:nu-def}
\end{equation}
If $\supp(z)=\{j_1<\cdots<j_t\}$ and
\begin{equation}
        g_0^L=0,
        \qquad
        g_k^L=g_{k-1}^L+\ell_{j_k}^L,
\end{equation}
then repeated use of the Pauli multiplication rule gives
\begin{equation}
        \nu_L(z)=i^{\sum_{k=1}^t \mu_{n_L}(g_{k-1}^L,\ell_{j_k}^L)}.
        \label{eq:nu-recursion}
\end{equation}
This formula is useful both for proofs and for efficient verification.

\begin{lemma}[Reordering phase]
For every word $w=(i_1,\ldots,i_\ell)$ and $L\in\{A,B\}$,
\begin{equation}
        P_L(\ell_{i_1}^L)\cdots P_L(\ell_{i_\ell}^L)
        =\vartheta_\Sigma(w)\,\mathcal P_L(z(w)),
\end{equation}
where
\begin{equation}
        \vartheta_\Sigma(w)=(-1)^{N_\Sigma(w)},
        \qquad
        N_\Sigma(w)=\sum_{\substack{1\le p<q\le \ell\\ i_p>i_q}} \Sigma_{i_pi_q}
        \pmod 2.
\end{equation}
In particular, the source and target have the same word-reordering phase whenever \eqref{eq:graph-kernel-common-sigma} holds.
\end{lemma}

\begin{proof}
Start from the ordered product appearing in the word.  We transform it into canonical increasing order by a sequence of adjacent swaps.  If two adjacent generators have labels $i$ and $j$, one swap changes their order by
\[
        P_L(\ell_i^L)P_L(\ell_j^L)
        =(-1)^{\Sigma_{ij}}P_L(\ell_j^L)P_L(\ell_i^L).
\]
Thus the only information needed for reordering is the anticommutation matrix $\Sigma_L$.  An inversion of the word is a pair $p<q$ with $i_p>i_q$; exactly one crossing is needed for each such pair.  Multiplying the signs over all inversions gives $\vartheta_\Sigma(w)$.

After the word is sorted, all copies of a fixed generator are adjacent.  Because each generator squares to identity, an even number of copies cancels completely and an odd number leaves one copy.  Hence the surviving generator set is exactly the parity vector $z(w)$, in increasing order.  This proves the formula.  If $\Sigma_A=\Sigma_B$, the same inversion product occurs in the source and target, so the word-reordering phase is common to both Hamiltonians.
\end{proof}

Combining the lemma with \eqref{eq:nu-def},
\begin{equation}
        P_L(\ell_{i_1}^L)\cdots P_L(\ell_{i_\ell}^L)
        =\vartheta_\Sigma(w)\nu_L(z(w))P_L(Lz(w)).
        \label{eq:word-product}
\end{equation}

\subsection{Canonical phase ratio}

Write
\begin{equation}
        \frac{\beta_i}{\alpha_i}=(-1)^{r_i},
        \qquad r_i\in\F_2.
\end{equation}
For a word with parity vector $z$, the ratio of the target and source sign factors is
\begin{equation}
        \prod_t \frac{\beta_{i_t}}{\alpha_{i_t}}=(-1)^{r\cdot z}.
\end{equation}
Because the reordering phase $\vartheta_\Sigma(w)$ is common under $\Sigma_A=\Sigma_B$, the full source-to-target phase ratio for a word with parity vector $z$ is
\begin{equation}
        \xi(z):=\frac{\nu_B(z)}{\nu_A(z)}(-1)^{r\cdot z}.
        \label{eq:xi-def}
\end{equation}
At this stage $\xi$ is a function of $z$.  Transduction in the compressed source basis requires it to be constant on cosets $z+K_A$.

\subsection{Signed relation phases}

Assume $K_A\subseteq K_B$.  For $s\in K_A$, define the ordered signed relation products
\begin{equation}
        R_A(s):=\prod_{i=1}^m(\alpha_iP_A(a_i))^{s_i}=\zeta_A(s)I,
\end{equation}
\begin{equation}
        R_B(s):=\prod_{i=1}^m(\beta_iP_B(b_i))^{s_i}=\zeta_B(s)I.
\end{equation}
The scalar phases are
\begin{equation}
        \zeta_A(s)=\nu_A(s)\prod_i \alpha_i^{s_i},
        \qquad
        \zeta_B(s)=\nu_B(s)\prod_i \beta_i^{s_i}.
\end{equation}

\begin{definition}[Signed-relation compatibility]
\label{def:signed-rel-compatible}
The pair $(H_A,H_B)$ is signed-relation compatible if
\begin{equation}
        \zeta_A(s)=\zeta_B(s)
        \qquad\text{for every }s\in K_A.
        \label{eq:signed-rel-compatible}
\end{equation}
Equivalently,
\begin{equation}
        \frac{\nu_B(s)}{\nu_A(s)}(-1)^{r\cdot s}=1
        \qquad\text{for every }s\in K_A.
        \label{eq:signed-rel-equivalent}
\end{equation}
\end{definition}

This condition says that every relation already present in the source is represented by the same signed scalar relation in the target.

\subsection{Why signed-relation compatibility is exactly the needed condition}

For $z,s\in\F_2^m$, define
\begin{equation}
        \kappa_\Sigma(z,s):=(-1)^{\sum_{i>j}z_is_j\Sigma_{ij}}.
\end{equation}
The exponent counts the crossings needed to move the generators selected by $s$ into canonical order among those selected by $z$.

\begin{lemma}[Canonical-product multiplication]
For $L\in\{A,B\}$,
\begin{equation}
        \mathcal P_L(z)\mathcal P_L(s)=\kappa_\Sigma(z,s)\mathcal P_L(z+s).
        \label{eq:canonical-mult}
\end{equation}
\end{lemma}

\begin{proof}
The product $\mathcal P_L(z)\mathcal P_L(s)$ is the concatenation of two increasing-order words.  To restore global increasing order, every generator $j$ selected by $s$ crosses every generator $i>j$ selected by $z$.  Each crossing contributes $(-1)^{\Sigma_{ij}}$.  Duplicate generators then cancel in pairs.
\end{proof}

If $s\in K_A\subseteq K_B$, then $Ls=0$ for $L=A,B$.  Using \eqref{eq:nu-def} and \eqref{eq:canonical-mult},
\begin{equation}
        \nu_L(z)\nu_L(s)P_L(Lz)=\kappa_\Sigma(z,s)\nu_L(z+s)P_L(Lz),
\end{equation}
and hence
\begin{equation}
        \nu_L(z+s)=\kappa_\Sigma(z,s)^{-1}\nu_L(z)\nu_L(s).
\end{equation}
Taking the target/source ratio cancels the common $\kappa_\Sigma$ factor:
\begin{equation}
        \frac{\nu_B(z+s)}{\nu_A(z+s)}
        =\frac{\nu_B(z)}{\nu_A(z)}\frac{\nu_B(s)}{\nu_A(s)}.
\end{equation}
Therefore
\begin{equation}
        \xi(z+s)=\xi(z)\frac{\nu_B(s)}{\nu_A(s)}(-1)^{r\cdot s}.
        \label{eq:xi-coset}
\end{equation}

\begin{lemma}[Well-defined phase correction]
\label{lem:chi-well-defined}
Assume $\Sigma_A=\Sigma_B$ and $K_A\subseteq K_B$.  The function
\begin{equation}
        \chi(Az):=\frac{\nu_B(z)}{\nu_A(z)}(-1)^{r\cdot z}
        \label{eq:chi-def}
\end{equation}
is independent of the representative $z$ if and only if signed-relation compatibility holds.
\end{lemma}

\begin{proof}
If $Az=A(z+s)$ with $s\in K_A$, then \eqref{eq:xi-coset} and signed-relation compatibility give $\xi(z+s)=\xi(z)$, so $\chi$ is well defined.  Conversely, compare the two representatives $z=0$ and $z=s\in K_A$ of the zero source label.  Well-definedness requires $\xi(s)=\xi(0)=1$, which is exactly \eqref{eq:signed-rel-equivalent}.
\end{proof}

\subsection{Basis test for signed-relation compatibility}

\begin{proposition}[Basis test]\label{prop:basis-phase-test}
Let $s^{(1)},\ldots,s^{(k)}$ be a basis of $K_A$.  Under $\Sigma_A=\Sigma_B$ and $K_A\subseteq K_B$, signed-relation compatibility holds for every $s\in K_A$ if and only if
\begin{equation}
        \zeta_A(s^{(j)})=\zeta_B(s^{(j)}),
        \qquad j=1,\ldots,k.
\end{equation}
\end{proposition}

\begin{proof}
Define the relative relation phase
\[
        \Delta(s):=\frac{\zeta_B(s)}{\zeta_A(s)}.
\]
We first check that $\Delta$ is multiplicative on $K_A$.  For $s,t\in K_A$, multiply the ordered products for $s$ and $t$.  In both source and target, the same reordering sign is needed to convert the concatenated product into the canonical product for $s+t$, because the anticommutation matrices agree.  Therefore the crossing phase cancels in the ratio between target and source, and
\[
        \Delta(s+t)=\Delta(s)\Delta(t).
\]
Thus $\Delta$ is a character of the binary vector space $K_A$.

If signed-relation compatibility holds for every relation, it holds on any basis.  Conversely, suppose $\Delta(s^{(j)})=1$ for each basis vector.  Any relation $s\in K_A$ has a unique expansion $s=\sum_j a_j s^{(j)}$ with $a_j\in\F_2$.  By multiplicativity,
\[
        \Delta(s)=\prod_j \Delta(s^{(j)})^{a_j}=1.
\]
Hence $\zeta_A(s)=\zeta_B(s)$ for all $s\in K_A$.
\end{proof}

\subsection{Cost of evaluating the phase correction \texorpdfstring{$\chi(Az)$}{chi(Az)}}
\label{subsec:chi-efficient}
After algebraic compatibility has been verified, the phase
\(\chi(u)\) depends only on the source label \(u\), even though it is
computed using a representative \(z\) satisfying \(Az=u\). It can be
evaluated coherently as follows.

\begin{enumerate}
    \item \emph{Choose a representative.}
    During classical preprocessing, fix a linear section
    \(S_A:\operatorname{im}A\to\mathbb F_2^m\) satisfying
    \(AS_Au=u\). On input \(u\), compute
    \begin{equation}
        z_0=S_Au.
    \end{equation}

    \item \emph{Compute the source and target Pauli phases.}
    For each \(L\in\{A,B\}\), initialize a running Pauli label
    \(g_L=0\) and a two-bit accumulator \(e_L=0\). Scan the indices
    \(i=1,\ldots,m\). Whenever \(z_{0,i}=1\), update
    \begin{equation}
        e_L
        \leftarrow
        e_L+\mu_{n_L}(g_L,\ell_i^L)
        \pmod 4,
        \qquad
        g_L\leftarrow g_L+\ell_i^L,
    \end{equation}
    where \(\ell_i^L\) is the \(i\)th Pauli label and
    \(\mu_{n_L}\) is the Pauli multiplication exponent. At the end,
    \begin{equation}
        \nu_L(z_0)=i^{e_L}.
    \end{equation}

    \item \emph{Compute the coefficient-sign contribution.}
    Define \(r\in\mathbb F_2^m\) by
    \begin{equation}
        \frac{\beta_i}{\alpha_i}=(-1)^{r_i}.
    \end{equation}
    Compute the parity
    \begin{equation}
        p=r\cdot z_0\pmod 2.
    \end{equation}

    \item \emph{Combine the contributions.}
    The phase exponent is
    \begin{equation}
        \kappa_\chi(u)
        =
        e_B-e_A+2p
        \pmod 4,
    \end{equation}
    and hence
    \begin{equation}
        \chi(u)
        =
        i^{\kappa_\chi(u)}
        =
        \frac{\nu_B(z_0)}{\nu_A(z_0)}
        (-1)^{r\cdot z_0}.
    \end{equation}

\end{enumerate}

The procedure uses only reversible binary linear algebra, Pauli phase
arithmetic, and a parity computation. For dense Pauli labels its gate
cost is \(O(m(n_A+n_B))\). It does not require a minimum-weight
representative or a call to the relative decoder.

\subsection{Explicit coefficient matching}
\label{subsec:explicit-coefficient-matching}

We prove that the phase-corrected source and target coefficients agree term by term.

Expanding $Q_D(H_L)$ using Eq.~\eqref{eq:word-product} gives
\begin{equation}
    Q_D(H_L)
    =\sum_{\ell=0}^D q_\ell
      \sum_{w=(i_1,\ldots,i_\ell)\in[m]^\ell}
      \left(\prod_{t=1}^\ell \lambda_{i_t}\right)
      \left(\prod_{t=1}^\ell \gamma^L_{i_t}\right)
      \vartheta_\Sigma(w)\nu_L(z(w))P_L(Lz(w)),
      \label{eq:polynomial-expansion}
\end{equation}
where $\gamma_i^A=\alpha_i$ and $\gamma_i^B=\beta_i$.  For $u\in\im A$, define the source aggregate coefficient
\begin{equation}
        C_A(u)
        :=\sum_{\ell=0}^D q_\ell
        \sum_{\substack{w=(i_1,\ldots,i_\ell)\in[m]^\ell\\ Az(w)=u}}
        \left(\prod_{t=1}^\ell \lambda_{i_t}\right)
        \left(\prod_{t=1}^\ell \alpha_{i_t}\right)
        \vartheta_\Sigma(w)\nu_A(z(w)).
        \label{eq:source-aggregate-coefficient}
\end{equation}
Similarly, for $v\in\im B$, define
\begin{equation}
        C_B(v)
        :=\sum_{\ell=0}^D q_\ell
        \sum_{\substack{w=(i_1,\ldots,i_\ell)\in[m]^\ell\\ Bz(w)=v}}
        \left(\prod_{t=1}^\ell \lambda_{i_t}\right)
        \left(\prod_{t=1}^\ell \beta_{i_t}\right)
        \vartheta_\Sigma(w)\nu_B(z(w)).
        \label{eq:target-aggregate-coefficient}
\end{equation}
The Bell-label states are then
\begin{equation}
        \ket{\widetilde\Psi_A^Q}=\frac{1}{N_A}\sum_{u\in \mathcal S_D(A)}C_A(u)\ket{u},
        \qquad
        \ket{\widetilde\Psi_B^Q}=\frac{1}{N_B}\sum_{v\in M(\mathcal S_D(A))}C_B(v)\ket{v},
        \label{eq:bell-label-states-coefficients}
\end{equation}
where the target indexing set is $M(\mathcal S_D(A))$: every degree-$D$ target word has the form $v=Bz=M(Az)$ with $Az\in\mathcal S_D(A)$. Some coefficients can vanish, so these are support envelopes rather than assertions of nonzero support.

\begin{proposition}[Coefficient matching]
\label{prop:coefficient-matching}
Assume $\Sigma_A=\Sigma_B$, $K_A\subseteq K_B$, signed-relation compatibility, and injectivity of $M$ on $\mathcal S_D(A)$.  Then
\begin{equation}
        C_B(Mu)=\chi(u)C_A(u),
        \qquad u\in \mathcal S_D(A).
\end{equation}
Consequently, $N_A=N_B$; when $Q_D(H_A)\ne0$, the normalized states satisfy
\begin{equation}
        \ket{\widetilde\Psi_B^Q}
        =\frac{1}{N_A}\sum_{u\in \mathcal S_D(A)}\chi(u)C_A(u)\ket{Mu}.
\end{equation}
\end{proposition}

\begin{proof}
We write out the two aggregate coefficients explicitly.  For a word $w=(i_1,\ldots,i_\ell)$, define
\[
        \Lambda(w):=\prod_{t=1}^\ell \lambda_{i_t},
        \qquad
        \gamma_L(w):=\prod_{t=1}^\ell \gamma^L_{i_t},
        \qquad
        z=z(w).
\]
By the expansion \eqref{eq:polynomial-expansion}, the source coefficient of label $u$ is
\begin{equation}
        C_A(u)
        =\sum_{\ell=0}^D q_\ell
          \sum_{\substack{w\in[m]^\ell\\ Az(w)=u}}
          \Lambda(w)\gamma_A(w)\vartheta_\Sigma(w)\nu_A(z(w)).
        \label{eq:CA-explicit}
\end{equation}
Similarly, the target coefficient of label $Mu$ is
\begin{equation}
        C_B(Mu)
        =\sum_{\ell=0}^D q_\ell
          \sum_{\substack{w\in[m]^\ell\\ Bz(w)=Mu}}
          \Lambda(w)\gamma_B(w)\vartheta_\Sigma(w)\nu_B(z(w)).
        \label{eq:CB-explicit}
\end{equation}
The two outer sums use the same polynomial coefficients, and the same word $w$ has the same magnitude factor $\Lambda(w)$ in source and target.  The nontrivial points are therefore: which words appear in the inner sums, and what is the phase ratio for each common word.

First compare the word sets.  If a word appears in the source sum, then $Az(w)=u$, so by definition of $M$,
\[
        Bz(w)=M(Az(w))=Mu.
\]
Thus every source word also appears in the target sum.  Conversely, suppose a word appears in the target sum, so $Bz(w)=Mu$.  Because $\ell\le D$, we have $Az(w)\in \mathcal S_D(A)$.  Also
\[
        M(Az(w))=Bz(w)=Mu.
\]
Since $M$ is injective on $\mathcal S_D(A)$ and $u\in \mathcal S_D(A)$, this implies $Az(w)=u$.  Hence every target word also appears in the source sum.  The two inner sums therefore run over exactly the same words.

Now fix one of these common words.  Since $Az(w)=u$, the well-defined phase correction gives
\[
        \chi(u)=\frac{\nu_B(z(w))}{\nu_A(z(w))}(-1)^{r\cdot z(w)}.
\]
The sign ratio of the Hamiltonian coefficients is
\[
        \frac{\gamma_B(w)}{\gamma_A(w)}
        =\prod_{t=1}^\ell \frac{\beta_{i_t}}{\alpha_{i_t}}
        =(-1)^{r\cdot z(w)}.
\]
Therefore
\[
        \gamma_B(w)\nu_B(z(w))
        =\chi(u)\,\gamma_A(w)\nu_A(z(w)).
\]
All other factors in \eqref{eq:CA-explicit} and \eqref{eq:CB-explicit} are identical.  Thus every target summand is $\chi(u)$ times the corresponding source summand, proving $C_B(Mu)=\chi(u)C_A(u)$.

Finally, $M$ is injective on $\mathcal S_D(A)$, so the target labels $Mu$ are distinct as $u$ ranges over $\mathcal S_D(A)$.  Since $|\chi(u)|=1$,
\[
        N_B^2=\sum_{u\in \mathcal S_D(A)}|C_B(Mu)|^2
        =\sum_{u\in \mathcal S_D(A)}|\chi(u)C_A(u)|^2
        =\sum_{u\in \mathcal S_D(A)}|C_A(u)|^2=N_A^2.
\]
The displayed expression for $\ket{\widetilde\Psi_B^Q}$ follows after normalization.
\end{proof}

\section{Kernel, graph, and their relations}
\label{app:graph-kernel-independence}

\subsection{Definitions and local notation}
\label{subsec:graph-kernel-symbols}
\label{subsec:anticommutation-graph-pairwise}
\label{def:projective-pauli-independence}

Let
\begin{equation}
        L=[\ell_1^L,\ldots,\ell_m^L]
        \in\F_2^{2n_L\times m}
        \label{eq:graph-kernel-label-matrix}
\end{equation}
be the binary label matrix of an indexed family of $m$ Pauli operators on $n_L$ qubits.  A vector $c\in\F_2^m$ is a \emph{term-selection vector}; the projective Pauli label of the selected product is $Lc$.  We use the following spaces:
\begin{equation}
\begin{aligned}
        W_L&:=\im L,
        &K_L&:=\ker L,
        \label{eq:relation-kernel-detailed}
\end{aligned}
\end{equation}
\begin{equation}
\begin{aligned}
        \Sigma_L&:=L^TJ_{n_L}L,
        &K_{\Sigma_L}&:=\ker\Sigma_L,
        \label{eq:graph-kernel-sigma-definition}
\end{aligned}
\end{equation}
\begin{equation}
\begin{aligned}
        W_L^{\perp_J}
        &:=\{g\in\F_2^{2n_L}:g^TJ_{n_L}w=0\text{ for all }w\in W_L\},
        \label{eq:graph-kernel-orthogonal-definition}
\end{aligned}
\end{equation}
\begin{equation}
\begin{aligned}
        \rad(W_L)&:=W_L\cap W_L^{\perp_J}.
        \label{eq:graph-kernel-radical-definition}
\end{aligned}
\end{equation}
The matrix $\Sigma_L$ is the adjacency matrix of the anticommutation graph on the \emph{term indices}: its $(i,j)$ entry equals one exactly when the $i$th and $j$th Pauli terms anticommute.  The relation kernel $K_L$ records higher-order multiplicative dependencies: $c\in K_L$ exactly when the selected product has zero projective Pauli label and hence is a scalar fourth root of unity.

\subsection{Summary of results}

\begin{proposition}
\label{prop:graph-kernel-central-products}
For every $c\in\F_2^m$, the following are equivalent:
\begin{enumerate}[label=(\roman*)]
    \item $c\in K_{\Sigma_L}$;
    \item the selected Pauli product commutes with every indexed generator;
    \item $Lc\in\rad(W_L)$.
\end{enumerate}
Consequently,
\begin{equation}
        K_L\subseteq K_{\Sigma_L}.
        \label{eq:kernel-in-graph-kernel-detailed}
\end{equation}
\end{proposition}

\begin{proof}
Fix $c\in\F_2^m$.  The projective label of the selected product is $Lc$.  Its commutation sign with the $j$th indexed generator is determined by
\begin{equation}
\begin{aligned}
        (\ell_j^L)^TJ_{n_L}(Lc)
        &=e_j^TL^TJ_{n_L}Lc \\
        &=e_j^T\Sigma_Lc.
        \label{eq:central-product-row-check}
\end{aligned}
\end{equation}
The selected product commutes with the $j$th generator exactly when the $j$th entry of $\Sigma_Lc$ vanishes.  It commutes with every indexed generator exactly when
\[
        \Sigma_Lc=0,
\]
which is condition (i).  This proves (i)$\Leftrightarrow$(ii).

The vector $Lc$ always belongs to $W_L=\im L$.  Equation~\eqref{eq:central-product-row-check} vanishes for every $j$ exactly when $Lc$ is symplectically orthogonal to every column of $L$, and therefore to every vector in their span $W_L$.  Thus
\[
        \Sigma_Lc=0
        \quad\Longleftrightarrow\quad
        Lc\in W_L\cap W_L^{\perp_J}
        =\rad(W_L),
\]
which proves (i)$\Leftrightarrow$(iii).

Finally, if $c\in K_L$, then $Lc=0$, and therefore
\[
        \Sigma_Lc=L^TJ_{n_L}Lc=0.
\]
Hence $c\in K_{\Sigma_L}$, proving the inclusion.
\end{proof}

Selections in $K_L$ yield scalar products, whereas selections in
$K_{\Sigma_L}\setminus K_L$ yield nonidentity Pauli operators that commute
with every generated operator. These central labels can become relations
without changing the pairwise commutation structure.

\begin{proposition}
\label{thm:graph-kernel-exact-sequence}
The map
\begin{equation}
        \phi_L:K_{\Sigma_L}\longrightarrow\rad(W_L),
        \qquad
        \phi_L(c):=Lc,
        \label{eq:graph-kernel-phi-map}
\end{equation}
is surjective and has kernel $K_L$. Hence,
\begin{equation}
        K_{\Sigma_L}/K_L
        \cong
        \rad(W_L),
        \label{eq:graph-kernel-radical-isomorphism}
\end{equation}
and 
\begin{equation}
        \dim K_{\Sigma_L}
        =\dim K_L+\dim\rad(W_L).
        \label{eq:short-exact-dimension}
\end{equation}
\end{proposition}

\begin{proof}
\emph{Step 1: well-definedness.} If $c\in K_{\Sigma_L}$, then Proposition~\ref{prop:graph-kernel-central-products} gives $Lc\in\rad(W_L)$.  Thus the codomain in Eq.~\eqref{eq:graph-kernel-phi-map} is correct.

\emph{Step 2: the kernel.} By definition,
\[
\begin{aligned}
        \ker\phi_L
        &=\{c\in K_{\Sigma_L}:Lc=0\}\\
        &=K_L,
\end{aligned}
\]
where the second equality again uses $K_L\subseteq K_{\Sigma_L}$.  

\emph{Step 3: surjectivity.}
Take any $g\in\rad(W_L)$.  Since $g\in W_L=\im L$, choose $c\in\F_2^m$ satisfying
\[
        g=Lc.
\]
Since $g\in W_L^{\perp_J}$, it is orthogonal to every column of $L$, so
\[
        L^TJ_{n_L}g=0.
\]
Substituting $g=Lc$ gives
\[
        \Sigma_Lc=L^TJ_{n_L}Lc=0.
\]
Thus $c\in K_{\Sigma_L}$ and $\phi_L(c)=g$. 

The first isomorphism theorem now gives
\[
        K_{\Sigma_L}/\ker\phi_L
        \cong\im\phi_L.
\]
Substituting $\ker\phi_L=K_L$ and $\im\phi_L=\rad(W_L)$ yields Eq.~\eqref{eq:graph-kernel-radical-isomorphism}. Finally, by the rank–nullity theorem, Eq.~\eqref{eq:short-exact-dimension} is proved.
\end{proof}
In particular, $K_L=K_{\Sigma_L}$ is equivalent to $\rad(W_L)=0$.

\begin{corollary}
\label{cor:graph-kernel-dimension-accounting}
For every indexed Pauli family,
\begin{equation}
\begin{aligned}
        \dim K_L&=m-\rank L,
        \label{eq:dim-KL}
\end{aligned}
\end{equation}
\begin{equation}
\begin{aligned}
        \dim K_{\Sigma_L}&=m-\rank\Sigma_L,
        \label{eq:dim-KSigma}
\end{aligned}
\end{equation}
\begin{equation}
\begin{aligned}
        \dim\rad(W_L)&=\rank L-\rank\Sigma_L.
        \label{eq:dim-radical-detailed}
\end{aligned}
\end{equation}
In particular, $\rank\Sigma_L\le\rank L$, and $\rank\Sigma_L$ is even.  If
\begin{equation}
        \rank\Sigma_L=2q,
        \qquad
        s:=\dim\rad(W_L),
\end{equation}
then
\begin{equation}
        \rank L=2q+s,
        \qquad
        \dim K_L=m-(2q+s).
        \label{eq:q-s-decomposition}
\end{equation}
\end{corollary}

\begin{proof}
The first two identities are rank-nullity theorem applied to $L$ and $\Sigma_L$.  Taking Eq.~\eqref{eq:short-exact-dimension} gives
\[
\begin{aligned}
        \dim\rad(W_L)
        &=(m-\rank\Sigma_L)-(m-\rank L)\\
        &=\rank L-\rank\Sigma_L.
\end{aligned}
\]
Nonnegativity gives $\rank\Sigma_L\le\rank L$.  The matrix $\Sigma_L$ represents an alternating bilinear form, and alternating forms have even rank.  Writing this rank as $2q$ gives the final decomposition.
\end{proof}

The $2q$ noncentral directions form $q$ independent anticommuting
symplectic pairs. The $s$ radical directions are independent non-scalar products that commute with
the entire family. The remaining $m- (2q + s)$ degrees of redundancy are scalar relations.

\begin{corollary}
For a commuting family, $\Sigma_L=0$, so $K_{\Sigma_L}=\F_2^m$ and $\rad(W_L)=W_L$.   
\end{corollary}

\begin{proof}
Since $\Sigma_L=0$, $K_{\Sigma_L}=\ker \Sigma_L=\ker0=\F_2^m$.

Take arbitrary \(u,v\in W_L\). Since \(W_L=\operatorname{im}L\), there
exist \(a,b\in\mathbb{F}_2^m\) such that $u = La$, $v = Lb$. It follows that
\[
\begin{aligned}
    u^{T}J_{n_L}v
    &= (La)^{T}J_{n_L}(Lb) \\
    &= a^{T}L^{T}J_{n_L}Lb \\
    &= a^{T}\Sigma_L b \\
    &= 0,
\end{aligned}
\]
where the last equality follows from \(\Sigma_L=0\). Thus every vector in \(W_L\) is symplectically orthogonal to every
vector in \(W_L\). Equivalently,
\[
    W_L\subseteq W_L^{\perp_J}.
\]
By the definition of the symplectic radical, $\operatorname{rad}(W_L):=W_L\cap W_L^{\perp_J}$,  $\operatorname{rad}(W_L)= W_L$ is proved.
\end{proof}

\begin{proposition}[Prescribed kernel under a fixed anticommutation matrix]
\label{thm:prescribed-target-kernel}
Let
\begin{equation}
    \Sigma\in\mathbb{F}_2^{m\times m}
\end{equation}
be a symmetric matrix with zero diagonal, and let
\begin{equation}
    K\subseteq\mathbb{F}_2^m
\end{equation}
be a linear subspace. Then there exist an integer \(n\) and a binary
Pauli-label matrix
\begin{equation}
    L\in\mathbb{F}_2^{2n\times m}
\end{equation}
such that
\begin{equation}
    \ker L=K,
    \qquad
    L^T J_n L=\Sigma,
\end{equation}
if and only if
\begin{equation}
    K\subseteq\ker\Sigma.
\end{equation}

Moreover, a concrete construction of $L$ enables $n\leq m$ and 
\begin{equation}
    n=m
    \quad\Longleftrightarrow\quad
    \Sigma=0
    \text{ and }
    K=\{0\}.
\end{equation}
\end{proposition}

\begin{proof}
We first prove necessity. Suppose that a matrix \(L\) satisfying
\[
    \ker L=K,
    \qquad
    L^T J_nL=\Sigma
\]
exists. For every \(c\in K\), we have \(Lc=0\), and hence
\[
\begin{aligned}
    \Sigma c
    &=
    L^T J_n Lc \\
    &=
    0.
\end{aligned}
\]
Therefore,
\[
    K\subseteq\ker\Sigma.
\]

For sufficiency, assume
\[
    K\subseteq\ker\Sigma.
\]
Consider the quotient space
\[
    \overline V
    :=
    \mathbb{F}_2^m/K,
\]
and let
\[
    \pi:\mathbb{F}_2^m\longrightarrow\overline V,
    \qquad
    \pi(c)=c+K
\]
be the quotient map.

We define a bilinear form on \(\overline V\) by
\begin{equation}
    \overline\omega(c+K,d+K)
    :=
    c^T\Sigma d.
    \label{eq:quotient-alternating-form}
\end{equation}
This form is well defined. Indeed, if \(k,k'\in K\), then
\[
\begin{aligned}
    (c+k)^T\Sigma(d+k')
    &=
    c^T\Sigma d
    +
    c^T\Sigma k'
    +
    k^T\Sigma d
    +
    k^T\Sigma k' \\
    &=
    c^T\Sigma d,
\end{aligned}
\]
because
\[
    \Sigma k=\Sigma k'=0.
\]
Thus changing either representative by an element of \(K\) does not
change the value of \(\overline\omega\).

The radical of the quotient form is
\[
    \operatorname{rad}(\overline V)
    =
    \ker\Sigma/K.
\]
To see this, observe that
\[
\begin{aligned}
    c+K\in\operatorname{rad}(\overline V)
    &\Longleftrightarrow
    \overline\omega(c+K,d+K)=0
    \quad
    \text{for every }d\in\mathbb{F}_2^m
    \\
    &\Longleftrightarrow
    c^T\Sigma d=0
    \quad
    \text{for every }d\in\mathbb{F}_2^m
    \\
    &\Longleftrightarrow
    \Sigma c=0
    \\
    &\Longleftrightarrow
    c\in\ker\Sigma.
\end{aligned}
\]

Write
\[
    \operatorname{rank}\Sigma=2q
\]
and
\[
    s
    :=
    \dim\ker\Sigma-\dim K.
\]
The rank is even because \(\Sigma\) defines an alternating bilinear
form. By the standard canonical-form theorem for alternating bilinear
spaces, there is a basis of \(\overline V\) of the form
\[
    p_1,q_1,\ldots,p_q,q_q,r_1,\ldots,r_s
\]
such that
\[
    \overline\omega(p_i,q_j)=\delta_{ij},
\]
while
\[
\begin{aligned}
    \overline\omega(p_i,p_j)
    &=
    0,\\
    \overline\omega(q_i,q_j)
    &=
    0,
\end{aligned}
\]
and
\[
    \overline\omega(r_a,v)=0
    \qquad
    \text{for every }v\in\overline V.
\]
Thus \(p_i,q_i\) form \(q\) symplectic pairs, whereas
\(r_1,\ldots,r_s\) span the radical.

Set
\[
    n:=q+s,
\]
and let
\[
    e_1,\ldots,e_n
\]
be the standard basis of \(\mathbb{F}_2^n\). Define a linear map
\[
    T:\overline V\longrightarrow\mathbb{F}_2^{2n}
\]
on the chosen basis by
\begin{equation}
\begin{aligned}
    T(p_i)
    &:=
    (e_i,0),
    &&
    i=1,\ldots,q,
    \label{eq:embed-p}
\end{aligned}
\end{equation}
\begin{equation}
\begin{aligned}
    T(q_i)
    &:=
    (0,e_i),
    &&
    i=1,\ldots,q,
    \label{eq:embed-q}
\end{aligned}
\end{equation}
\begin{equation}
\begin{aligned}
    T(r_a)
    &:=
    (0,e_{q+a}),
    &&
    a=1,\ldots,s.
    \label{eq:embed-radical}
\end{aligned}
\end{equation}
The vectors in
Eqs.~\eqref{eq:embed-p}--\eqref{eq:embed-radical} are linearly
independent, so \(T\) is injective.

Moreover, \(T\) preserves the alternating form. Indeed,
\[
\begin{aligned}
    T(p_i)^T J_n T(q_j)
    &=
    (e_i,0)^T
    \begin{pmatrix}
        0&I_n\\
        I_n&0
    \end{pmatrix}
    (0,e_j)
    \\
    &=
    e_i^Te_j
    \\
    &=
    \delta_{ij},
\end{aligned}
\]
while all other pairings among the images of
\(p_i,q_i,r_a\) vanish. Therefore, by bilinearity,
\begin{equation}
    T(v)^TJ_nT(w)
    =
    \overline\omega(v,w)
    \qquad
    \text{for every }v,w\in\overline V.
    \label{eq:T-isometry}
\end{equation}

Finally, define
\[
    L
    :=
    T\circ\pi:
    \mathbb{F}_2^m
    \longrightarrow
    \mathbb{F}_2^{2n}.
\]
Because \(T\) is injective,
\[
\begin{aligned}
    Lc=0
    &\Longleftrightarrow
    T(\pi(c))=0
    \\
    &\Longleftrightarrow
    \pi(c)=0
    \\
    &\Longleftrightarrow
    c\in K.
\end{aligned}
\]
Hence
\[
    \ker L=K.
\]

It remains to verify the anticommutation matrix. For arbitrary
\(c,d\in\mathbb{F}_2^m\), Eqs.~\eqref{eq:T-isometry} and
\eqref{eq:quotient-alternating-form} give
\[
\begin{aligned}
    (Lc)^TJ_n(Ld)
    &=
    T(\pi(c))^TJ_nT(\pi(d))
    \\
    &=
    \overline\omega(\pi(c),\pi(d))
    \\
    &=
    c^T\Sigma d.
\end{aligned}
\]
Since this equality holds for every \(c,d\in\mathbb{F}_2^m\), we
conclude that
\[
    L^TJ_nL=\Sigma.
\]
This completes the construction.

To show $n\leq m$, let
\[
    k := \dim K.
\]
Since $\Sigma \in \mathbb{F}_2^{m\times m}$ and
$\operatorname{rank}\Sigma = 2q$, the rank--nullity theorem gives
\[
\begin{aligned}
    \dim\ker\Sigma
    &= m-\operatorname{rank}\Sigma \\
    &= m-2q.
\end{aligned}
\]
By the definition of $s$,
\[
\begin{aligned}
    s
    &= \dim\ker\Sigma-\dim K \\
    &= m-2q-k.
\end{aligned}
\]
Substituting this expression into $n=q+s$, we obtain
\[
\begin{aligned}
    n
    &= q+s \\
    &= q+(m-2q-k) \\
    &= m-q-k.
\end{aligned}
\]
Equivalently,
\[
\begin{aligned}
    m-n
    &= q+k \\
    &= \frac{1}{2}\operatorname{rank}\Sigma+\dim K.
\end{aligned}
\]
Because $q$ and $k$ are nonnegative integers,
\[
    m-n \geq 0.
\]
Therefore,
\[
    n\leq m.
\]

Equality holds if and only if
\[
\begin{aligned}
    n=m
    &\Longleftrightarrow m-n=0 \\
    &\Longleftrightarrow q+k=0 \\
    &\Longleftrightarrow q=0
       \text{ and }
       k=0.
\end{aligned}
\]
Since
\[
\begin{aligned}
    q=0
    &\Longleftrightarrow \operatorname{rank}\Sigma=0
     \Longleftrightarrow \Sigma=0,
\end{aligned}
\]
and
\[
    k=0
    \Longleftrightarrow
    \dim K=0
    \Longleftrightarrow
    K=\{0\},
\]
we conclude that
\[
    n=m
    \quad\Longleftrightarrow\quad
    \Sigma=0
    \text{ and }
    K=\{0\}.
\]
Thus equality occurs precisely when the indexed Pauli family is both commuting and relation free.

In every other case,
\[
    \operatorname{rank}\Sigma>0
    \quad\text{or}\quad
    \dim K>0,
\]
so that
\[
    n<m.
\]
More explicitly, the reduction in the number of qubits is
\[
    m-n
    =
    \frac{1}{2}\operatorname{rank}\Sigma+\dim K.
\]
\end{proof}

\begin{proposition}[Universal relation-free realization]
\label{prop:universal-pauli-graph}
Every symmetric zero-diagonal matrix $\Sigma\in\F_2^{m\times m}$ has a projectively independent Pauli realization.  Explicitly, on $m$ qubits define
\begin{equation}
        \widehat P_i
        :=X_i\prod_{j<i}Z_j^{\Sigma_{ji}}.
        \label{eq:universal-pauli-graph-construction}
\end{equation}
Then the anticommutation matrix is $\Sigma$ and the relation kernel is zero.
\end{proposition}

\begin{proof}
The binary label of $\widehat P_i$ is $(e_i,z_i)$ with
\[
        z_i=\sum_{j<i}\Sigma_{ji}e_j.
\]
For $i<j$,
\[
        (e_i,z_i)^TJ_m(e_j,z_j)
        =e_i\cdot z_j+z_i\cdot e_j
        =\Sigma_{ij}.
\]
Thus the pairwise anticommutation data are correct.  The $X$ half of the label matrix is the identity $I_m$, so the columns are linearly independent and the kernel is zero.
\end{proof}

This relation-free realization supplies the source Hamiltonian underlying HDQI.

\subsection{The source-to-target map is a central quotient}
\label{subsec:graph-kernel-transduction-consequences}

\begin{proposition}[The induced map is a central quotient]
\label{thm:induced-map-central-quotient}
Assume
\begin{equation}
        \Sigma_A=A^TJ_{n_A}A
        =B^TJ_{n_B}B=:\Sigma,
        \qquad
        K_A\subseteq K_B.
        \label{eq:graph-kernel-common-sigma}
\end{equation}
Define
\begin{equation}
        M:W_A\longrightarrow W_B,
        \qquad
        M(Az):=Bz.
        \label{eq:graph-kernel-induced-map}
\end{equation}
Then we have the map $M$ is well defined and surjective.  Moreover:
\begin{enumerate}[label=(\roman*)]
    \item $M$ preserves the symplectic pairing:
    \begin{equation}
        (Mu)^TJ_{n_B}(Mv)
        =u^TJ_{n_A}v
        \quad\text{for all }u,v\in W_A;
        \label{eq:M-preserves-pairing}
    \end{equation}
    \item the map
    \begin{equation}
        \psi:K_B/K_A\longrightarrow\ker M,
        \qquad
        \psi(c+K_A):=Ac,
        \label{eq:ker-M-quotient-code}
    \end{equation}
    is an isomorphism, and $\ker M\subseteq\rad(W_A)$;
    \item $M(\rad(W_A))=\rad(W_B)$ and the induced map
    \begin{equation}
        \overline M:
        W_A/\rad(W_A)
        \longrightarrow
        W_B/\rad(W_B),
        \qquad
        \overline M(u+\rad(W_A))
        :=Mu+\rad(W_B),
        \label{eq:nondegenerate-core-isomorphism}
    \end{equation}
    is a symplectic isomorphism;
    \item
    \begin{equation}
        \dim K_B-\dim K_A
        =\dim\ker M
        =\rank A-\rank B
        =\dim\rad(W_A)-\dim\rad(W_B).
        \label{eq:independence-loss-identities}
    \end{equation}
\end{enumerate}
\end{proposition}

\begin{proof}
\emph{Step 1: well-definedness and surjectivity.}
If $Az=Az'$, then $z+z'\in K_A\subseteq K_B$, so $Bz=Bz'$.  Hence $M(Az)$ does not depend on the representative.  Every $v\in W_B$ has the form $v=Bz=M(Az)$, so $M$ is surjective.

\emph{Step 2: preservation of the symplectic pairing.}
Write $u=Ac$ and $v=Ac'$.  Then
\[
\begin{aligned}
        (Mu)^TJ_{n_B}(Mv)
        &=(Bc)^TJ_{n_B}(Bc')\\
        &=c^TB^TJ_{n_B}Bc'\\
        &=c^T\Sigma c'\\
        &=c^TA^TJ_{n_A}Ac'\\
        &=u^TJ_{n_A}v.
\end{aligned}
\]

\emph{Step 3: identify the kernel with the quotient code.}
First, $\psi$ is well defined: if $c+K_A=c'+K_A$, then $c+c'\in K_A$, so $Ac=Ac'$.  If $c\in K_B$, then
\[
        M(Ac)=Bc=0,
\]
so the image lies in $\ker M$.  Conversely, if $u=Ac\in\ker M$, then $Bc=M(Ac)=0$, so $c\in K_B$ and $u=\psi(c+K_A)$.  Finally, $\psi(c+K_A)=0$ exactly when $Ac=0$, or $c\in K_A$, so the quotient class is zero.  Hence $\psi$ is an isomorphism.

If $u\in\ker M$, then for every $v\in W_A$,
\[
        u^TJ_{n_A}v
        =(Mu)^TJ_{n_B}(Mv)=0.
\]
Since $u\in W_A$, this proves $u\in\rad(W_A)$.  Thus every direction collapsed by the target map is central in the source.

\emph{Step 4: map the radicals onto each other.}
Take $r_A\in\rad(W_A)$.  For every $y\in W_B$, choose $v\in W_A$ with $Mv=y$.  Then
\[
        (Mr_A)^TJ_{n_B}y
        =(Mr_A)^TJ_{n_B}(Mv)
        =r_A^TJ_{n_A}v=0.
\]
Thus $Mr_A\in\rad(W_B)$, proving $M(\rad(W_A))\subseteq\rad(W_B)$.

Conversely, take $r_B\in\rad(W_B)$ and choose $u\in W_A$ with $Mu=r_B$.  For every $v\in W_A$,
\[
        u^TJ_{n_A}v
        =(Mu)^TJ_{n_B}(Mv)
        =r_B^TJ_{n_B}(Mv)=0.
\]
Hence $u\in\rad(W_A)$ and $r_B=Mu$.  Therefore
\begin{equation}
        M(\rad(W_A))=\rad(W_B).
        \label{eq:M-radical-surjective}
\end{equation}

\emph{Step 5: prove the quotient map is well defined and injective.}
Equation~\eqref{eq:M-radical-surjective} makes $\overline M$ well defined and surjective.  To prove injectivity, suppose
\[
        \overline M(u+\rad(W_A))
        =\rad(W_B).
\]
This means $Mu\in\rad(W_B)$.  For every $v\in W_A$,
\[
        u^TJ_{n_A}v
        =(Mu)^TJ_{n_B}(Mv)=0,
\]
because $Mu$ is radical.  Since $u\in W_A$, this implies $u\in\rad(W_A)$.  Therefore $u+\rad(W_A)$ is the zero quotient class, and $\ker\overline M=0$.  The map is injective.  Pairing preservation descends to the quotients, so it is a symplectic isomorphism.

\emph{Step 6: dimensions.}
The quotient isomorphism gives
\[
        \dim\ker M=\dim K_B-\dim K_A.
\]
Rank--nullity for the surjective map $M:W_A\to W_B$ gives
\[
        \dim\ker M=\dim W_A-\dim W_B=\rank A-\rank B.
\]
Finally, Corollary~\ref{cor:graph-kernel-dimension-accounting} applied to $A$ and $B$ with the same $\rank\Sigma$ gives the radical-dimension identity.
\end{proof}

Source and target have isomorphic nondegenerate symplectic quotients.
The map removes only central source labels, whose independent directions
are described by $K_B/K_A$; scalar phases are treated separately.

\section{Quotient codes, relative decoding, distance design, and the mapping hierarchy}
\label{app:relative}\label{sec:relative-decoding}\label{app:mapping}

\subsection{Quotient code and decoding goal}

When $K_A\subseteq K_B$, the map $B$ descends to
\begin{equation}
        B:\F_2^m/K_A\to \im B,
        \qquad
        B(z+K_A)=Bz.
\end{equation}
Its kernel is $K_B/K_A$.  The decoder should recover
\begin{equation}
        Bz\longmapsto Az,
        \label{eq:relative-decoding-goal}
\end{equation}
which is equivalent to recovering the class $z+K_A$, not the original representative $z$.

For $u\in\im A$, define the source-generator weight
\begin{equation}
        w_A(u):=\min_{z:Az=u}|z|.
\end{equation}
If $u=Az$, then all representatives of $u$ are $z+s$ with $s\in K_A$, so
\begin{equation}
        w_A(Az)=\min_{s\in K_A}|z+s|.
\end{equation}
The degree-$D$ support is therefore
\begin{equation}
        \mathcal S_D(A)=\{Az: |z|\le D\}=\{u\in\im A:w_A(u)\le D\}.
\end{equation}

\subsection{Relative distance and injectivity radius}

\begin{definition}[Relative distance]
The relative distance is
\begin{equation}
        d_{A\to B}:=
        \min_{\substack{0\ne u\in\im A\\ Mu=0}}w_A(u),
\end{equation}
with $d_{A\to B}=+\infty$ if $M$ is injective on all of $\im A$.
\end{definition}

This definition is intrinsic to the source-label space: it asks for the lightest nonzero source label that disappears after applying $M$.  The following equivalent formula is often easier to use because it is written directly in terms of the relation codes.

\begin{proposition}[Equivalent code formula]
\label{prop:relative-distance-code-formula}
Assuming $K_A\subseteq K_B$,
\begin{equation}
        d_{A\to B}=\min_{c\in K_B\setminus K_A}|c|,
\end{equation}
with $d_{A\to B}=+\infty$ when $K_A=K_B$.
\end{proposition}

\begin{proof}
Take any source label $u\in\im A$ and choose a representative $c\in\F_2^m$ with $u=Ac$.  The condition that $u$ is killed by $M$ is
\[
        Mu=0
        \quad\Longleftrightarrow\quad
        M(Ac)=0
        \quad\Longleftrightarrow\quad
        Bc=0
        \quad\Longleftrightarrow\quad
        c\in K_B.
\]
The condition that the source label is nonzero is
\[
        u\ne0
        \quad\Longleftrightarrow\quad
        Ac\ne0
        \quad\Longleftrightarrow\quad
        c\notin K_A.
\]
Thus nonzero labels in $\ker M$ are represented precisely by vectors in $K_B\setminus K_A$.

The only remaining issue is the weight.  The same source label $Ac$ has all representatives $c+s$ with $s\in K_A$, so
\[
        w_A(Ac)=\min_{s\in K_A}|c+s|.
\]
If $c\in K_B\setminus K_A$, then $c+s$ remains in $K_B\setminus K_A$ for every $s\in K_A$.  Conversely, every element of $K_B\setminus K_A$ appears in this way by taking $s=0$.  Therefore minimizing the quotient weight over all nonzero labels in $\ker M$ is the same as minimizing the ordinary Hamming weight over $K_B\setminus K_A$.
\end{proof}

\begin{proposition}[Injectivity radius]
\label{thm:injectivity-radius}
The map $M$ is injective on $\mathcal S_D(A)$ if and only if
\begin{equation}
        2D<d_{A\to B}.
\end{equation}
\end{proposition}

\begin{proof}
We prove both directions.

First assume $2D<d_{A\to B}$.  Let $u,u'\in \mathcal S_D(A)$ satisfy $Mu=Mu'$.  By definition of $\mathcal S_D(A)$, choose representatives $z,z'$ with
\[
        u=Az,
        \qquad
        u'=Az',
        \qquad
        |z|\le D,
        \qquad
        |z'|\le D.
\]
If $u\ne u'$, then $u+u'$ is a nonzero element of $\ker M$, because
\[
        M(u+u')=Mu+Mu'=0.
\]
Its source-generator weight is at most the weight of the explicit representative $z+z'$:
\[
        w_A(u+u')=w_A(A(z+z'))\le |z+z'|\le |z|+|z'|\le 2D.
\]
This contradicts the definition of $d_{A\to B}$ as the minimum weight of a nonzero element of $\ker M$.  Hence $u=u'$, and $M$ is injective on $\mathcal S_D(A)$.

Conversely, suppose $d_{A\to B}\le 2D$.  By the code formula, there is a vector $c\in K_B\setminus K_A$ with $|c|\le2D$.  Split the support of $c$ into two disjoint subsets, each of size at most $D$, and write
\[
        c=z_1+z_2,
        \qquad
        |z_1|\le D,
        \qquad
        |z_2|\le D.
\]
Set $u_1=Az_1$ and $u_2=Az_2$.  Then $u_1,u_2\in \mathcal S_D(A)$.  They are distinct because
\[
        u_1+u_2=A(z_1+z_2)=Ac\ne0,
\]
where $Ac\ne0$ follows from $c\notin K_A$.  However their target labels agree:
\[
        Mu_1+Mu_2=M(u_1+u_2)=M(Ac)=Bc=0,
\]
where $Bc=0$ follows from $c\in K_B$.  Thus $Mu_1=Mu_2$ with $u_1\ne u_2$, so $M$ is not injective on $\mathcal S_D(A)$.
\end{proof}

\subsection{Source kernel affects relative distance}
\label{sec: kernel affect distance}

The ordinary target-code distance is
\begin{equation}
        d_B:=\min_{0\ne c\in K_B}|c|.
\end{equation}
Since $K_B\setminus K_A\subseteq K_B\setminus\{0\}$, one always has $d_{A\to B}\ge d_B$ when both are finite.

\begin{proposition}[Strict improvement]
\label{prop:strict-improvement}
If $d_B<\infty$, then
\begin{equation}
        d_{A\to B}>d_B
        \quad\Longleftrightarrow\quad
        \{c\in K_B: |c|=d_B\}\subseteq K_A.
\end{equation}
\end{proposition}

\begin{proof}
If a minimum-weight target relation $c$ is not in $K_A$, then $c\in K_B\setminus K_A$ and $d_{A\to B}\le |c|=d_B$.  Together with $d_{A\to B}\ge d_B$, this gives equality.  Conversely, if every weight-$d_B$ target relation lies in $K_A$, then every element of $K_B\setminus K_A$ has weight strictly larger than $d_B$.
\end{proof}

\begin{property}[Residual distance is monotone under absorbing source relations]
\label{prop:relative-distance-monotonicity}
If
\begin{equation}
        K_A\subseteq K_A'\subseteq K_B,
\end{equation}
then
\begin{equation}
        d_{A'\to B}\ge d_{A\to B}.
        \label{eq:relative-distance-monotonicity}
\end{equation}
If $K_A'=K_B$, the distance jumps to $+\infty$.
\end{property}

\begin{proof}
The residual sets satisfy
\[
        K_B\setminus K_A'
        \subseteq
        K_B\setminus K_A.
\]
Taking a minimum over a smaller set cannot decrease the result.  When the two kernels are equal, the residual set is empty by definition.
\end{proof}

Enlarging the source kernel cannot decrease the relative distance.
Exploiting a larger degree requires efficient source preparation and
relative decoding.

\subsection{Fixed source: the largest achievable residual distance}
\label{subsec:linear-relative-distance-existence}

Assume a common anticommutation matrix $\Sigma$ and let $K_\Sigma:=\ker\Sigma$.  Appendix~\ref{app:graph-kernel-independence} proves that every compatible relation kernel lies between $K_A$ and $K_\Sigma$, and that every intermediate subspace is realizable at the unrestricted binary-Pauli level.  For $c\in K_\Sigma$, define the distance of its coset from $K_A$ by
\begin{equation}
        \operatorname{dist}(c,K_A)
        :=\min_{s\in K_A}|c+s|.
        \label{eq:extension-radius}
\end{equation}
Define the extension radius
\begin{equation}
        R_{\Sigma/A}
        :=\max_{c\in K_\Sigma}\operatorname{dist}(c,K_A).
        \label{eq:far-radical-coset}
\end{equation}

\begin{property}[Fixed source: exact optimum over compatible target kernels]
\label{prop:max-achievable-relative-distance}
If $K_A\subsetneq K_\Sigma$, then
\begin{equation}
        \max_{K_A\subsetneq K_B\subseteq K_\Sigma}
        d_{A\to B}
        =R_{\Sigma/A}.
        \label{eq:max-distance-equals-covering-radius}
\end{equation}
A one-dimensional extension
\begin{equation}
        K_B=K_A\oplus\Span\{c_\star\}
        \label{eq:one-dimensional-best-extension}
\end{equation}
attains the optimum whenever $c_\star+K_A$ is a farthest coset.
\end{property}

\begin{proof}
For any admissible target kernel,
\[
        d_{A\to B}
        =\min_{\substack{[c]\in K_B/K_A\\ {}[c]\ne0}}
          \operatorname{dist}(c,K_A).
\]
Every quotient class in $K_B/K_A$ is also a class in $K_\Sigma/K_A$, so this minimum cannot exceed the largest admissible coset distance $R_{\Sigma/A}$.

Choose $c_\star$ attaining the maximum and define the one-dimensional extension~\eqref{eq:one-dimensional-best-extension}.  The quotient $K_B/K_A$ then has exactly one nonzero class, $c_\star+K_A$.  Its quotient weight is $\operatorname{dist}(c_\star,K_A)=R_{\Sigma/A}$, so the upper bound is attained.
\end{proof}

For a source family, a target with $d_{A\to B}=\Theta(m)$ exists at the unrestricted binary-Pauli level exactly when $R_{\Sigma/A}=\Theta(m)$. 

\subsection{Fixed target: which source kernel achieves a threshold?}
\label{subsec:fixed-target-source-design}

For an integer threshold $t$, define
\begin{equation}
        K_B^{<t}
        :=\Span\{c\in K_B:|c|<t\}.
        \label{eq:fixed-target-low-weight-span}
\end{equation}

\begin{property}[Fixed target: exact low-weight-span criterion]
\label{prop:fixed-target-source-criterion}
There exists a proper source kernel $K_A\subsetneq K_B$ satisfying
\begin{equation}
        d_{A\to B}\ge t
        \label{eq:fixed-target-distance-threshold}
\end{equation}
if and only if
\begin{equation}
        K_B^{<t}\subsetneq K_B.
        \label{eq:fixed-target-span-proper}
\end{equation}
Whenever the condition holds, any subspace satisfying
\begin{equation}
        K_B^{<t}\subseteq K_A\subsetneq K_B
        \label{eq:fixed-target-optimal-distance}
\end{equation}
has residual distance at least $t$.
\end{property}

\begin{proof}
\emph{Necessity.}
If $d_{A\to B}\ge t$, every target relation of weight below $t$ must lie in $K_A$; otherwise it would be a residual relation below the threshold.  Because $K_A$ is a subspace, it contains the span of all such relations:
\[
        K_B^{<t}\subseteq K_A\subsetneq K_B.
\]
Thus the low-weight span is proper.

\emph{Sufficiency.}
Assume $K_B^{<t}\subsetneq K_B$ and choose any proper intermediate subspace~\eqref{eq:fixed-target-optimal-distance}; a hyperplane containing $K_B^{<t}$ always exists.  If $c\in K_B\setminus K_A$ had weight below $t$, then by definition $c\in K_B^{<t}\subseteq K_A$, a contradiction.  Hence every residual relation has weight at least $t$.
\end{proof}

\subsection{Mapping hierarchy: unitary equivalence and relation-changing maps}
\label{sec:mapping-hierarchy}
\label{subsec:mapping-hierarchy-summary}
\label{subsec:clifford-infinite-distance}

\begin{proposition}[Direct unitary conjugation transports every degree]
\label{prop:direct-unitary-all-degrees}
If $H_B=UH_AU^\dagger$, then for every polynomial $Q$,
\begin{equation}
        Q(H_B)=UQ(H_A)U^\dagger,
        \label{eq:unitary-polynomial-functional-calculus}
\end{equation}
and, on equal-dimensional doubled systems,
\begin{equation}
        \ket{\Psi_B^Q}
        =(U\otimes U^*)\ket{\Psi_A^Q}.
        \label{eq:normalized-unitary-choi-map}
\end{equation}
The Gibbs states obey $\rho_\beta(H_B)=U\rho_\beta(H_A)U^\dagger$.
\end{proposition}

\begin{proof}
For every integer $\ell\ge0$,
\[
        H_B^\ell=(UH_AU^\dagger)^\ell=UH_A^\ell U^\dagger,
\]
because adjacent factors $U^\dagger U$ cancel.  Linearity gives the polynomial identity, and functional calculus gives the exponential identity.  For the Choi vector, use
\[
        (U^\dagger\otimes I)\ket{\Phi_n}
        =(I\otimes U^*)\ket{\Phi_n}
\]
to move the rightmost $U^\dagger$ to the reference register.  Unitary conjugation preserves the normalization.
\end{proof}

The operator identity $H_B=UH_AU^\dagger$ need not respect a chosen decomposition into indexed Pauli terms.  The relation-code hierarchy concerns the stronger termwise condition.

\begin{property}[Termwise unitary equivalence is the infinite-distance endpoint]
\label{prop:clifford-infinite-distance}
\label{cor:finite-distance-not-unitary}
Let
\begin{equation}
        \widetilde P_i^A:=\alpha_iP_A(a_i),
        \qquad
        \widetilde P_i^B:=\beta_iP_B(b_i).
        \label{eq:signed-pauli-generators-hierarchy}
\end{equation}
If a unitary satisfies
\begin{equation}
        \widetilde P_i^B
        =U\widetilde P_i^AU^\dagger
        \quad\text{for every }i,
        \label{eq:termwise-unitary-equivalence}
\end{equation}
then
\begin{equation}
        K_A=K_B,
        \qquad
        d_{A\to B}=+\infty,
        \label{eq:clifford-relative-distance-infinity}
\end{equation}
and every signed relation phase is preserved.  If $U$ is Clifford, there is an invertible symplectic matrix $S$ such that
\begin{equation}
        B=SA,
        \qquad
        M(Az)=S(Az).
        \label{eq:clifford-symplectic-label-map}
\end{equation}
Conversely,
\begin{equation}
        d_{A\to B}<\infty
        \quad\Longrightarrow\quad
        K_A\subsetneq K_B
        \quad\Longrightarrow\quad
        \text{no termwise unitary conjugation exists}.
        \label{eq:finite-distance-nontrivial-certificate}
\end{equation}
Every finite residual distance belongs to $\{1,\ldots,m\}$.
\end{property}

\begin{proof}
For a selection vector $c$, define the ordered signed products
\[
        R_L(c):=\prod_{i=1}^m(\widetilde P_i^L)^{c_i}.
\]
Termwise conjugation gives
\begin{equation}
        R_B(c)=UR_A(c)U^\dagger.
        \label{eq:unitary-relation-product-map}
\end{equation}
A unitary maps a scalar operator to the same scalar and a non-scalar operator to a non-scalar operator.  Thus $R_A(c)$ is scalar if and only if $R_B(c)$ is scalar, proving $K_A=K_B$.  If $R_A(c)=\zeta I$, then $R_B(c)=\zeta I$, so signed phases agree.  The residual set is empty, giving infinite distance.

For a Clifford, conjugation induces an invertible symplectic transformation on binary Pauli labels, giving $B=SA$ after matching the indexed terms.

Conversely, if the distance is finite, choose $c\in K_B\setminus K_A$.  The selected target product is scalar, while the selected source product has nonzero Pauli label $Ac$ and is non-scalar.  Equation~\eqref{eq:unitary-relation-product-map} is impossible.  Finally, every nonzero binary vector has Hamming weight between one and $m$.
\end{proof}

Finite relative distance implies a new target relation and therefore
rules out termwise unitary conjugation.

\subsection{About relative decoder}
\label{subsec:standard-vs-relative}

\begin{proposition}[Ordinary decoders as relative decoders]
\label{thm:ordinary-decoder-relative-decoder}
Let
\begin{equation}
    A\in\mathbb{F}_2^{r_A\times m},
    \qquad
    B\in\mathbb{F}_2^{r_B\times m},
\end{equation}
and assume
\begin{equation}
    K_A:=\ker A
    \subsetneq
    K_B:=\ker B.
\end{equation}
Define the ordinary target distance and the relative distance by
\begin{equation}
    d_B
    :=
    \min_{0\neq c\in K_B}|c|,
    \qquad
    d_{A\to B}
    :=
    \min_{c\in K_B\setminus K_A}|c|,
\end{equation}
Fix a degree bound
\(D\ge 0\).

\begin{enumerate}[label=(\roman*)]
    \item Suppose an ordinary target decoder
    \begin{equation}
        \operatorname{Dec}_B:\operatorname{im}B
        \longrightarrow \mathbb{F}_2^m
    \end{equation}
    satisfies
    \begin{equation}
        \operatorname{Dec}_B(Bz)=z
        \qquad
        \text{for every }z\in\mathbb{F}_2^m
        \text{ with }|z|\le D.
    \end{equation}
    Then
    \begin{equation}
        \operatorname{Dec}_{A\to B}^{\mathrm{std}}(y)
        :=
        A\operatorname{Dec}_B(y)
        \label{eq:standard-to-relative-decoder}
    \end{equation}
    is a valid relative decoder:
    \begin{equation}
        \operatorname{Dec}_{A\to B}^{\mathrm{std}}(Bz)
        =
        Az
        \qquad
        \text{for every }|z|\le D.
    \end{equation}
    In particular, an ordinary unique decoder certified by the target
    distance works whenever
    \begin{equation}
        2D<d_B.
    \end{equation}
    This construction does not exploit a possible separation
    \(d_{A\to B}>d_B\).

    \item For \(y\in\operatorname{im}B\), choose any minimum-weight
    representative
    \begin{equation}
        z^\star(y)
        \in
        \arg\min_{z':\,Bz'=y}|z'|,
        \label{eq:minimum-weight-representative}
    \end{equation}
    and define an ordinary minimum-weight decoder
    \begin{equation}
        \operatorname{Dec}_{A\to B}^{\mathrm{MW}}(y)
        :=
        Az^\star(y).
        \label{eq:minimum-weight-relative-decoder}
    \end{equation}
    If
    \begin{equation}
        2D<d_{A\to B},
        \label{eq:minimum-weight-relative-condition}
    \end{equation}
    then, for every \(z_0\in\mathbb{F}_2^m\) satisfying
    \(|z_0|\le D\),
    \begin{equation}
        \operatorname{Dec}_{A\to B}^{\mathrm{MW}}(Bz_0)
        =
        Az_0.
        \label{eq:minimum-weight-relative-correctness}
    \end{equation}
    Hence every minimum-weight representative of the fiber
    \(\{z':Bz'=Bz_0\}\) has the same source label \(Az_0\).
\end{enumerate}
\end{proposition}

\begin{proof}
For part~\emph{(i)}, let \(z\in\mathbb{F}_2^m\) satisfy \(|z|\le D\).
Using the assumed correctness of the ordinary decoder,
\[
\begin{aligned}
    \operatorname{Dec}_{A\to B}^{\mathrm{std}}(Bz)
    &=
    A\operatorname{Dec}_B(Bz)
    \\
    &=
    Az.
\end{aligned}
\]
Thus composing an ordinary decoder with the linear map \(A\) always
produces a valid relative decoder on the same promised error set.
However, the distance-based correctness guarantee for ordinary unique
decoding is only
\[
    2D<d_B.
\]
Therefore, this construction inherits the ordinary decoding radius and
does not use the potentially larger relative distance.

For part~\emph{(ii)}, let
\[
    y=Bz_0,
    \qquad
    |z_0|\le D,
\]
and let \(z^\star=z^\star(y)\) be any minimizer in
Eq.~\eqref{eq:minimum-weight-representative}. Since \(z_0\) is a
feasible representative of \(y\), minimum-weight optimality gives
\begin{equation}
    |z^\star|
    \le
    |z_0|
    \le
    D.
    \label{eq:minimizer-weight-bound}
\end{equation}

Because \(Bz^\star=Bz_0\), we have
\[
    B(z^\star+z_0)=0,
\]
and therefore
\begin{equation}
    z^\star+z_0\in K_B.
    \label{eq:difference-in-target-kernel}
\end{equation}
Suppose, for contradiction, that
\[
    Az^\star\neq Az_0.
\]
Then
\[
    A(z^\star+z_0)\neq0,
\]
so
\begin{equation}
    z^\star+z_0\notin K_A.
    \label{eq:difference-not-source-kernel}
\end{equation}
Combining
Eqs.~\eqref{eq:difference-in-target-kernel} and
\eqref{eq:difference-not-source-kernel} gives
\[
    z^\star+z_0\in K_B\setminus K_A.
\]
By the definition of the relative distance,
\[
\begin{aligned}
    d_{A\to B}
    &\le
    |z^\star+z_0|
    \\
    &\le
    |z^\star|+|z_0|
    \\
    &\le
    2D,
\end{aligned}
\]
where the last inequality follows from
Eq.~\eqref{eq:minimizer-weight-bound}. This contradicts
\(2D<d_{A\to B}\). Hence
\[
    Az^\star=Az_0,
\]
which proves
Eq.~\eqref{eq:minimum-weight-relative-correctness}.
\end{proof}

An ordinary bounded-distance decoder supplies relative decoding within
its certified radius through Eq.~\eqref{eq:standard-to-relative-decoder}.
A minimum-weight representative gives a valid relative decoder throughout
$2D<d_{A\to B}$ by Eq.~\eqref{eq:minimum-weight-relative-decoder}, but its
computation need not be efficient.

\section{Circuit workflow in detail}
\label{subsec:circuit-workflow}

The operational input to the transduction circuit is the physical polynomial-filtered Choi state, not an abstract Bell-label state.  Thus the first step is a Bell transform on the source physical-reference pair.  Recall
\begin{equation}
        \ket{\Psi_A^Q}_{A\bar A}
        =\frac{(Q_D(H_A)\otimes I)\ket{\Phi_{n_A}}}
        {\norm{(Q_D(H_A)\otimes I)\ket{\Phi_{n_A}}}}.
\end{equation}
Applying the source Bell transform gives the source label representation
\begin{equation}
        U_{\Bell,A}\ket{\Psi_A^Q}_{A\bar A}
        =\ket{\widetilde\Psi_A^Q}_S
        =\frac{1}{N_A}\sum_{u\in \mathcal S_D(A)}C_A(u)\ket{u}_S.
        \label{eq:source-bell-transform-start}
\end{equation}
Only after this Bell transform do we run the relative-decoding transduction.  Use three computational registers: the source Bell-label register $S$, a target Bell-label register $T$, and a recovery register $R$.  The target register is initially blank; at the end it is converted back to a physical target Choi state by the inverse target Bell transform.  The complete reversible workflow is
\begin{equation}
\begin{aligned}
    \ket{\Psi_A^Q}_{A\bar A}\ket{0}_T\ket{0}_R
    &\xrightarrow{\ U_{\Bell,A}\ }
    \frac{1}{N_A}\sum_{u\in \mathcal S_D(A)} C_A(u)\ket{u}_S\ket{0}_T\ket{0}_R,
    \label{eq:workflow-bell-start}
\end{aligned}
\end{equation}
\begin{equation}
\begin{aligned}
    &\xrightarrow{\ U_M\ }
    \frac{1}{N_A}\sum_{u} C_A(u)\ket{u}_S\ket{Mu}_T\ket{0}_R,
    \\
    &\xrightarrow{\ U_\chi\ }
    \frac{1}{N_A}\sum_{u} \chi(u)C_A(u)\ket{u}_S\ket{Mu}_T\ket{0}_R,
    \\
    &\xrightarrow{\ U_{\rm dec}\ }
    \frac{1}{N_A}\sum_{u} \chi(u)C_A(u)\ket{u}_S\ket{Mu}_T\ket{u}_R,
    \\
    &\xrightarrow{\ \operatorname{XOR}_{R\to S}\ }
    \frac{1}{N_A}\sum_{u} \chi(u)C_A(u)\ket{0}_S\ket{Mu}_T\ket{u}_R,
    \\
    &\xrightarrow{\ U_{\rm dec}^\dagger\ }
    \ket{0}_S\left(\frac{1}{N_A}\sum_{u} \chi(u)C_A(u)\ket{Mu}_T\right)\ket{0}_R,
    \\
    &\xrightarrow{\ U_{\Bell,B}^\dagger\ }
    \ket{0}_S\ket{\Psi_B^Q}_{B\bar B}\ket{0}_R.
    \label{eq:workflow-physical-end}
\end{aligned}
\end{equation}
Here $U_M$ computes $u\mapsto Mu$, $U_\chi$ applies the phase $\chi(u)$, and $U_{\rm dec}$ coherently recovers $u$ from the target label $Mu$.  The source register can be erased because the decoder supplies a second copy of $u$ in the recovery register.  After applying $U_{\rm dec}^\dagger$, only the target Bell-label superposition remains; by coefficient matching, this state is exactly $\ket{\widetilde\Psi_B^Q}$.  The final inverse Bell transform therefore prepares the physical target polynomial-filtered Choi state $\ket{\Psi_B^Q}$.


\section{DQI and HDQI pilot states as source polynomial states}
\label{app:dqi-hdqi-relation-free}

This appendix identifies the pilot states of DQI and HDQI with
polynomial-filtered states of explicit relation-free source
Hamiltonians. The result explains both constructions as special cases
of Hamiltonian transduction. We write
\begin{equation}
    Q_D(x)=\sum_{k=0}^{D}q_kx^k
\end{equation}
for a real polynomial of degree at most \(D\).

\subsection{The DQI pilot state}

The DQI source Hamiltonian is the product Hamiltonian
\begin{equation}
    H_{\rm sel}:=\sum_{i=1}^{m}Z_i.
\end{equation}
For \(x\in\mathbb F_2^m\),
\begin{equation}
    H_{\rm sel}\ket{x}=(m-2|x|)\ket{x}.
\end{equation}
Define
\begin{equation}
    f_Q(x):=Q_D(m-2|x|)
\end{equation}
and its Boolean Fourier coefficients
\begin{equation}
    \widehat f_Q(z)
    :=
    2^{-m}\sum_{x\in\mathbb F_2^m}
    f_Q(x)(-1)^{x\cdot z}.
\end{equation}

The unnormalized DQI pilot state is
\begin{equation}
    \ket{\widetilde R_Q^{\rm DQI}}
    :=
    \sum_{z\in\mathbb F_2^m}
    \widehat f_Q(z)\ket{z}.
    \label{eq:dqi-pilot}
\end{equation}
Because \(f_Q(x)\) depends only on \(|x|\), the coefficient
\(\widehat f_Q(z)\) depends only on \(|z|\). Moreover,
\(\widehat f_Q(z)=0\) whenever \(|z|>D\). Hence the pilot has the
weighted-Dicke form
\begin{equation}
    \ket{\widetilde R_Q^{\rm DQI}}
    =
    \sum_{k=0}^{D}
    \beta_k
    \sum_{|z|=k}\ket{z}
\end{equation}
for suitable real coefficients \(\beta_k\). This is the pilot state
used in DQI~\cite{JordanEtAlDQI}.

\begin{proposition}[DQI pilot as a source polynomial state]
\label{prop:dqi-pilot-source}
Assume $Q_D(H_{\rm sel})\ne0$.
Let
\begin{equation}
    \ket{\Psi_{\rm sel}^Q}
    :=
    \frac{(Q_D(H_{\rm sel})\otimes I)\ket{\Phi_m}}
    {\norm{(Q_D(H_{\rm sel})\otimes I)\ket{\Phi_m}}}.
\end{equation}
Let \(U_{\Bell,Z}\) denote the Bell-label isometry on the diagonal
Pauli Choi subspace (the zero $X$-label register is suppressed), defined by
\begin{equation}
    U_{\Bell,Z}
    (Z(z)\otimes I)\ket{\Phi_m}
    =
    \ket{z}.
\end{equation}
Then
\begin{equation}
    U_{\Bell,Z}\ket{\Psi_{\rm sel}^Q}
    =
    \frac{\ket{\widetilde R_Q^{\rm DQI}}}
    {\norm{\widetilde R_Q^{\rm DQI}}}.
\end{equation}
Equivalently, if \(V_{\rm copy}\ket{x}=\ket{x}\ket{x}\), then
\begin{equation}
    V_{\rm copy}H^{\otimes m}
    \ket{\widetilde R_Q^{\rm DQI}}
    =
    (Q_D(H_{\rm sel})\otimes I)\ket{\Phi_m}.
\end{equation}
\end{proposition}

\begin{proof}
Boolean Fourier inversion gives
\[
    f_Q(x)
    =
    \sum_{z\in\mathbb F_2^m}
    \widehat f_Q(z)(-1)^{x\cdot z}.
\]
Since
\[
    Z(z)\ket{x}=(-1)^{x\cdot z}\ket{x},
\]
we obtain the operator identity
\[
    Q_D(H_{\rm sel})
    =
    \sum_{z\in\mathbb F_2^m}
    \widehat f_Q(z)Z(z).
\]
Applying the Bell-label isometry gives
\[
\begin{aligned}
    U_{\Bell,Z}
    (Q_D(H_{\rm sel})\otimes I)\ket{\Phi_m}
    &=
    \sum_z
    \widehat f_Q(z)
    U_{\Bell,Z}(Z(z)\otimes I)\ket{\Phi_m} \\
    &=
    \sum_z\widehat f_Q(z)\ket{z}
     =
    \ket{\widetilde R_Q^{\rm DQI}}.
\end{aligned}
\]
Normalization proves the first identity.

For the second identity,
\[
\begin{aligned}
    H^{\otimes m}\ket{\widetilde R_Q^{\rm DQI}}
    &=
    2^{-m/2}
    \sum_x
    \left(
      \sum_z\widehat f_Q(z)(-1)^{x\cdot z}
    \right)\ket{x} \\
    &=
    2^{-m/2}\sum_xf_Q(x)\ket{x}.
\end{aligned}
\]
Applying \(V_{\rm copy}\) gives
\[
    2^{-m/2}\sum_xQ_D(m-2|x|)\ket{x}\ket{x}
    =
    (Q_D(H_{\rm sel})\otimes I)\ket{\Phi_m}.
\]
\end{proof}

\subsection{The HDQI pilot state}
Here $T$ is the triangular matrix of Eq.~\eqref{eq:main-triangular-sigma}.
The notation $R_Q^{\rm HDQI}(G)$ below denotes the same normalized pilot
as $\mathcal R_D(G)$ in Eq.~\eqref{eq:main-general-hdqi-pilot}.

For a general noncommuting target, let
\(\Sigma\in\mathbb F_2^{m\times m}\) be its anticommutation matrix.
On \(m\) source qubits define
\begin{equation}
    \widehat P_i
    :=
    X_i\prod_{j<i}Z_j^{\Sigma_{ji}},
    \qquad
    \widehat H_G
    :=
    \sum_{i=1}^{m}\widehat P_i.
    \label{eq:hdqi-relation-free-source}
\end{equation}
These generators satisfy
\begin{equation}
    \widehat P_i\widehat P_j
    =
    (-1)^{\Sigma_{ij}}
    \widehat P_j\widehat P_i.
\end{equation}
Their binary label matrix has the form
\begin{equation}
    \widehat A=
    \begin{pmatrix}
        I_m\\
        T
    \end{pmatrix},
\end{equation}
and therefore
\begin{equation}
    \ker\widehat A=\{0\}.
\end{equation}
Thus \(\widehat H_G\) is relation free and has the same
anticommutation graph as the target.

For a word
\begin{equation}
    w=(i_1,\ldots,i_k)\in[m]^k,
\end{equation}
let \(y(w)\in\mathbb F_2^m\) record the parity of the number of
occurrences of each index. Reordering the word into increasing order
produces the sign
\begin{equation}
    \vartheta_\Sigma(w)
    :=
    (-1)^{
      \sum_{\substack{p<q\\ i_p>i_q}}
      \Sigma_{i_pi_q}
    }.
\end{equation}
Define the HDQI beta functions
\begin{equation}
    \beta_G^{(k)}(y)
    :=
    \sum_{\substack{w\in[m]^k\\y(w)=y}}
    \vartheta_\Sigma(w)
\end{equation}
and
\begin{equation}
    a_G(y)
    :=
    \sum_{k=0}^{D}q_k\beta_G^{(k)}(y).
\end{equation}
The normalized HDQI pilot is
\begin{equation}
    \ket{R_Q^{\rm HDQI}(G)}
    :=
    \frac{1}{N_G}
    \sum_{y\in\mathbb F_2^m}
    a_G(y)\ket{y},
    \qquad
    N_G^2:=\sum_y|a_G(y)|^2.
    \label{eq:hdqi-pilot}
\end{equation}
This is the graph-dependent pilot used in
HDQI~\cite{HamiltonianDQI}. The complete HDQI initialization tensors
this pilot with a fresh maximally entangled target register.

\begin{proposition}[HDQI pilot as a source polynomial state]
\label{prop:hdqi-pilot-source}
Assume $Q_D(\widehat H_G)\ne0$ and define the normalized source polynomial state
\begin{equation}
    \ket{\Psi_G^Q}
    :=
    \frac{(Q_D(\widehat H_G)\otimes I)\ket{\Phi_m}}
    {\norm{(Q_D(\widehat H_G)\otimes I)\ket{\Phi_m}}}.
\end{equation}
There exists an efficiently implementable graph-adapted Bell-label unitary
\(U_G^{\rm lab}\) such that
\begin{equation}
    U_G^{\rm lab}\ket{\Psi_G^Q}
    =
    \ket{R_Q^{\rm HDQI}(G)}\ket{0^m}.
\end{equation}
\end{proposition}

\begin{proof}
For \(y\in\mathbb F_2^m\), define the increasing-order product
\[
    \widehat P^y
    :=
    \prod_{i=1}^{m}\widehat P_i^{y_i}.
\]
Expanding the \(k\)th power of \(\widehat H_G\) gives one operator word
for each \(w\in[m]^k\). Reordering that word into increasing order
contributes the sign \(\vartheta_\Sigma(w)\), while pairs of repeated
generators cancel because \(\widehat P_i^2=I\). Hence
\[
    \widehat H_G^k
    =
    \sum_{y\in\mathbb F_2^m}
    \beta_G^{(k)}(y)\widehat P^y.
\]
Therefore
\[
\begin{aligned}
    Q_D(\widehat H_G)
    &=
    \sum_{k=0}^{D}q_k\widehat H_G^k \\
    &=
    \sum_{y\in\mathbb F_2^m}
    a_G(y)\widehat P^y.
\end{aligned}
\]

Because \(\ker\widehat A=\{0\}\), distinct vectors \(y\) give distinct
Pauli labels. Consequently, the Choi vectors
\[
    (\widehat P^y\otimes I)\ket{\Phi_m},
    \qquad
    y\in\mathbb F_2^m,
\]
are mutually orthonormal. We may therefore choose a unitary
\(U_G^{\rm lab}\) satisfying
\[
    U_G^{\rm lab}
    (\widehat P^y\otimes I)\ket{\Phi_m}
    =
    \ket{y}\ket{0^m}.
\]
Applying it to the unnormalized source polynomial state yields
\[
\begin{aligned}
    U_G^{\rm lab}
    (Q_D(\widehat H_G)\otimes I)\ket{\Phi_m}
    &=
    \sum_y
    a_G(y)\ket{y}\ket{0^m} \\
    &=
    N_G\ket{R_Q^{\rm HDQI}(G)}\ket{0^m}.
\end{aligned}
\]
Orthogonality implies that the norm of the unnormalized source polynomial vector is
\(N_G\), so normalization proves the result.
\end{proof}


\section{Interaction complexity, locality, frustration, and classical simulability}
\label{app:dense}\label{sec:dense-target-completion}

This appendix asks what the algebraic invariants imply about the \emph{simplest physical realization} of a relation kernel.  The answer depends strongly on which notion of simplicity is optimized.  Kernel dimension counts independent dependencies; relation distance counts how many formal terms participate in the smallest dependency; Pauli weight, qubit participation, and geometric range describe a physical embedding; none of these parameters alone determines the others.

\subsection{Physical definitions}
\label{subsec:algebraic-versus-physical-locality}

For a Pauli label $\ell_i=(x_i,z_i)$, define its support and Pauli weight by
\begin{equation}
        S_i:=\{q:(x_{i,q},z_{i,q})\ne(0,0)\},
        \qquad
        \operatorname{wt}_P(\ell_i):=|S_i|.
        \label{eq:locality-support-and-weight}
\end{equation}
The number of inter-qubit interaction terms is
\begin{equation}
        t_{\rm int}(L)
        :=\#\{i:\operatorname{wt}_P(\ell_i)\ge2\}.
        \label{eq:appendix-interaction-count-properties}
\end{equation}
A family is \emph{reduced} when no column is zero and no two columns are equal.  This removes identity shifts and repeated Pauli terms, which can be combined in the Hamiltonian description.

For each physical qubit $q$, define
\begin{equation}
        d_q:=\#\{i:q\in S_i\},
        \qquad
        \Delta_{\rm part}:=\max_q d_q.
        \label{eq:locality-participation}
\end{equation}
If the qubits are embedded in a graph or metric space $G_{\rm phys}$, define
\begin{equation}
        R_{\rm loc}
        :=\max_i\operatorname{diam}_{G_{\rm phys}}(S_i).
        \label{eq:locality-range}
\end{equation}
For a prescribed kernel $K$, one may define the minimum possible maximum Pauli weight
\begin{equation}
        \ell_P(K)
        :=\min_{L:\ker L=K}\max_i\operatorname{wt}_P(\ell_i),
        \label{eq:minimal-pauli-locality-kernel}
\end{equation}
and, if a common anticommutation matrix is also prescribed,
\begin{equation}
        \ell_P(K,\Sigma)
        :=\min_{\substack{L:\ker L=K\\L^TJL=\Sigma}}
          \max_i\operatorname{wt}_P(\ell_i).
        \label{eq:minimal-pauli-locality-kernel-graph}
\end{equation}
These optimization parameters depend on the full code or binary-matroid structure, not only on $\dim K$.

\subsection{How many interactions are forced by a relation kernel?\label{sec: howmanyinter}}

\begin{property}[Systematic commuting realization: a product core plus at most $k$ interactions]
\label{prop:systematic-commuting-realization-expanded}
Let $K\subseteq\F_2^m$ have dimension $k$ and let $r=m-k$.  There is a commuting $Z$-type realization of $K$ on $r$ qubits containing $r$ one-qubit fields and at most $k$ additional parity interactions.  If
\begin{equation}
        d(K):=\min_{0\ne c\in K}|c|\ge3,
        \label{eq:kernel-distance-reducedness}
\end{equation}
the realization is reduced and contains exactly $k$ genuine interactions.
\end{property}

\begin{proof}
After permuting the $m$ term indices, choose a systematic generator matrix
\[
        G_K=\begin{pmatrix}I_k&P\end{pmatrix},
        \qquad
        P\in\F_2^{k\times r}.
\]
Define
\[
        L_{\rm sys}
        :=\begin{pmatrix}P^T&I_r\end{pmatrix}
        \in\F_2^{r\times m}.
\]
Write a selection vector as $z=(x,y)$ with $x\in\F_2^k$ and $y\in\F_2^r$.  Then
\[
        L_{\rm sys}z=P^Tx+y.
\]
Therefore
\[
        L_{\rm sys}z=0
        \quad\Longleftrightarrow\quad
        y=P^Tx,
\]
so
\[
        \ker L_{\rm sys}
        =\{(x,P^Tx):x\in\F_2^k\}
        =K.
\]
Interpret the columns as diagonal $Z$ labels.  The final $r$ identity columns are the one-qubit fields $Z_1,\ldots,Z_r$.  The first $k$ columns are the rows of $P$, interpreted as parity interactions.  Hence there are at most $k$ non-one-local terms.

If $d(K)\ge3$, no parity-check matrix with kernel $K$ can have a zero column, since that would create a weight-one codeword.  It cannot have two equal columns, since that would create a weight-two codeword.  In the systematic matrix, a first-block column of weight zero would be zero, while one of weight one would coincide with an identity column.  Thus all first $k$ columns have weight at least two, and all $k$ are genuine interactions.
\end{proof}

The construction uses at most $k$ interaction terms, each potentially
acting on up to $m-k$ qubits.

\begin{property}[Reduced commuting families require at least one interaction per independent relation]
\label{prop:commuting-interaction-lower-bound-expanded}
Let $L$ be a reduced commuting Pauli family and set $k:=\dim K_L$.  Then
\begin{equation}
        t_{\rm int}(L)\ge k.
        \label{eq:commuting-interaction-lower-bound-expanded}
\end{equation}
Moreover, after a Clifford change of basis, the family can be written, up to signs and spectator qubits, as $m-k$ independent one-qubit $Z$ fields plus exactly $k$ additional commuting $Z$-string interactions.  Consequently, for a fixed kernel $K$ with $d(K)\ge3$,
\begin{equation}
        \min_{\substack{L:\ker L=K\\L^TJL=0\\L\text{ reduced}}}
        t_{\rm int}(L)
        =\dim K.
        \label{eq:minimum-commuting-interactions-equals-kernel-dim}
\end{equation}
\end{property}

\begin{proof}
Let $S_1$ be the set of one-local terms.  In a reduced commuting family, these labels are linearly independent:
\begin{itemize}
    \item one-local terms on different qubits have disjoint binary support;
    \item distinct nonidentity Pauli axes on the same qubit anticommute, so at most one axis can occur in a commuting family;
    \item repeated labels are excluded by reducedness.
\end{itemize}
Hence
\[
        \rank L\ge|S_1|=m-t_{\rm int}(L).
\]
Rank--nullity gives
\[
        k=m-\rank L
        \le m-(m-t_{\rm int}(L))
        =t_{\rm int}(L).
\]

For the matching upper bound, choose $r=m-k$ linearly independent indexed generators.  Since they commute, symplectic Gaussian elimination produces a Clifford mapping them to $Z_1,\ldots,Z_r$.  Every remaining generator is, up to a sign, a binary product of this basis and hence becomes a signed $Z$ string on those $r$ qubits.  There are exactly $m-r=k$ remaining indexed terms.  Reducedness excludes the identity and excludes equality with a basis field, so every remaining string has weight at least two.  Thus the Clifford normal form contains exactly $k$ interaction terms.  Property~\ref{prop:systematic-commuting-realization-expanded} gives the same conclusion for a prescribed kernel.
\end{proof}

\subsection{The noncommuting generalization\label{sec:noncommuting}}

\begin{property}[Kernel dimension can be stored in a noncommuting product Hamiltonian]
\label{prop:noncommuting-product-counterexample-expanded}
For every $k$, the one-local family
\begin{equation}
        \{X_j,Y_j,Z_j:j=1,\ldots,k\}
        \label{eq:XYZ-product-family-expanded}
\end{equation}
has
\begin{equation}
        \dim K_L=k,
        \qquad
        t_{\rm int}(L)=0.
\end{equation}
\end{property}

\begin{proof}
On each qubit,
\[
        a_{X_j}+a_{Y_j}+a_{Z_j}=0,
\]
which gives one relation.  Different qubits have disjoint label support, so these relations are independent and there are no additional independent cross-qubit relations.  Every term is one local.
\end{proof}

Thus $\dim K$ alone cannot force inter-qubit interactions in a noncommuting realization.  Some relation dimension can be explained entirely by on-site $X/Y/Z$ triples.

\begin{property}[General lower bounds for noncommuting realizations]
\label{prop:noncommuting-interaction-lower-bounds-expanded}
Assume $L$ is reduced (all columns are nonzero and distinct). Let $L_1$ be the submatrix containing its one-local columns, and define
\begin{equation}
        k:=\dim\ker L,
        \qquad
        k_1:=\dim\ker L_1.
\end{equation}
Then
\begin{equation}
        t_{\rm int}(L)\ge k-k_1.
        \label{eq:k-minus-k1-bound-expanded}
\end{equation}
For a reduced one-local subfamily,
\begin{equation}
        k_1\le\frac12\rank\Sigma_L,
\end{equation}
so
\begin{equation}
        t_{\rm int}(L)
        \ge
        \max\left\{0,\dim K_L-\frac12\rank\Sigma_L\right\}.
        \label{eq:rankSigma-interaction-bound-expanded}
\end{equation}
There is also a graph-shape obstruction.  Let $\tau_{\rm 1loc}(\Sigma_L)$ be the minimum number of term vertices whose removal makes the anticommutation graph a disjoint union of cliques of sizes at most three.  Then
\begin{equation}
        t_{\rm int}(L)\ge\tau_{\rm 1loc}(\Sigma_L).
        \label{eq:graph-deletion-bound-expanded}
\end{equation}
\end{property}

\begin{proof}
Let $m_1$ be the number of one-local columns, so $m_1=m-t_{\rm int}(L)$.  Since $L_1$ is a submatrix of $L$,
\[
        \rank L\ge\rank L_1=m_1-k_1.
\]
Therefore
\[
\begin{aligned}
        k=m-\rank L
        &\le m-(m_1-k_1)\\
        &=t_{\rm int}(L)+k_1,
\end{aligned}
\]
which proves Eq.~\eqref{eq:k-minus-k1-bound-expanded}.

For the rank bound, inspect the reduced one-local family qubit by qubit.  One Pauli axis gives no local relation and an anticommutation block of rank zero.  Two distinct axes give no relation and a $2\times2$ anticommutation block of rank two.  All three axes give one relation $X+Y+Z=0$, while the $K_3$ anticommutation block still has rank two.  Different qubits give direct-sum blocks.  Hence every one-local relation contributes at least two to the rank of the one-local anticommutation matrix:
\[
        2k_1\le\rank\Sigma_{L_1}\le\rank\Sigma_L.
\]
Substitution gives Eq.~\eqref{eq:rankSigma-interaction-bound-expanded}.

Finally, remove all interaction-term vertices.  Every remaining term is one local.  Terms on different qubits commute; the one, two, or three distinct axes on a given qubit form a clique $K_1$, $K_2$, or $K_3$.  Hence the remaining graph is a disjoint union of cliques of size at most three.  At least $\tau_{\rm 1loc}(\Sigma_L)$ vertices must have been removed.
\end{proof}

The quantity $k_1$ counts independent relations within decoupled
one-qubit algebras; the remaining relation dimension requires
interaction terms.

\subsection{Signed relations and frustration}
\label{subsec:kernels-frustration-locality-hardness}

Consider
\begin{equation}
        H=\sum_i\lambda_i\gamma_iP_i,
        \qquad
        \lambda_i>0,\quad \gamma_i\in\{\pm1\},\quad P_i\ne I.
        \label{eq:frustration-hamiltonian-convention}
\end{equation}
Define the preferred constraint observables
\begin{equation}
        C_i:=-\gamma_iP_i,
        \qquad
        H=-\sum_i\lambda_iC_i.
        \label{eq:preferred-constraint-observables}
\end{equation}
Each term is individually minimized by $C_i=+1$.  For $c\in K_L$, write
\begin{equation}
        \prod_iC_i^{c_i}=\zeta_C(c)I.
        \label{eq:preferred-relation-phase}
\end{equation}

\begin{property}[Commuting frustration is exactly detected by signed relations]
\label{prop:commuting-pauli-frustration-criterion}
If the $C_i$ commute, then
\begin{equation}
        H\text{ is frustration free}
        \quad\Longleftrightarrow\quad
        \zeta_C(c)=+1
        \text{ for every }c\in K_L.
        \label{eq:commuting-frustration-free-condition-expanded}
\end{equation}
\end{property}

\begin{proof}
\noindent\emph{Necessity.}
Suppose that $H$ is frustration free. Then there exists a nonzero
state $|\psi\rangle$ such that
\[
    C_i|\psi\rangle=|\psi\rangle
    \qquad
    \text{for every }i.
\]
Let $c\in K_L$. By the definition of the relation kernel, the
corresponding product has zero projective Pauli label and is therefore
a scalar:
\[
    \prod_{i=1}^{m}C_i^{c_i}
    =
    \zeta_C(c)I.
\]
Applying both sides to $|\psi\rangle$ gives
\[
\begin{aligned}
    \zeta_C(c)|\psi\rangle
    &=
    \left(\prod_{i=1}^{m}C_i^{c_i}\right)|\psi\rangle \\
    &=
    |\psi\rangle.
\end{aligned}
\]
Since $|\psi\rangle\neq 0$, it follows that
\[
    \zeta_C(c)=+1.
\]
Therefore every signed relation must have positive phase.

\noindent\emph{Sufficiency.}
Suppose now that
\[
    \zeta_C(c)=+1
    \qquad
    \text{for every }c\in K_L.
\]
Let
\[
    \mathcal{S}
    :=
    \langle C_1,\ldots,C_m\rangle
\]
be the Pauli group generated by the preferred observables.

Because the $C_i$ commute and satisfy $C_i^2=I$, every element of
$\mathcal{S}$ can be written as
\[
    G(c)
    =
    \prod_{i=1}^{m}C_i^{c_i},
    \qquad
    c\in\mathbb{F}_2^m.
\]
Moreover, every such product is Hermitian, so a scalar element of
$\mathcal{S}$ can only be $+I$ or $-I$.

We first show that the assumed relation condition is equivalent to
\[
    -I\notin\mathcal{S}.
\]
Indeed, if $-I\in\mathcal{S}$, then for some
$c\in\mathbb{F}_2^m$,
\[
    \prod_{i=1}^{m}C_i^{c_i}=-I.
\]
Since the product is scalar, its projective Pauli label is zero, so
$c\in K_L$. This would imply
\[
    \zeta_C(c)=-1,
\]
contradicting the assumption. Conversely, a relation with
$\zeta_C(c)=-1$ immediately generates $-I$. Therefore,
\[
    \zeta_C(c)=+1
    \text{ for every }c\in K_L
    \quad\Longleftrightarrow\quad
    -I\notin\mathcal{S}.
\]

Define the group-averaging operator
\[
    \Pi_{\mathcal{S}}
    :=
    \frac{1}{|\mathcal{S}|}
    \sum_{G\in\mathcal{S}}G.
\]
This operator is a projector. Indeed,
\[
\begin{aligned}
    \Pi_{\mathcal{S}}^2
    &=
    \frac{1}{|\mathcal{S}|^2}
    \sum_{G,H\in\mathcal{S}}GH \\
    &=
    \frac{1}{|\mathcal{S}|}
    \sum_{K\in\mathcal{S}}K \\
    &=
    \Pi_{\mathcal{S}},
\end{aligned}
\]
because for each fixed $K\in\mathcal{S}$ there are exactly
$|\mathcal{S}|$ pairs $(G,H)$ satisfying $GH=K$.

Furthermore, left multiplication by any $C_i$ permutes the elements
of $\mathcal{S}$. Hence
\[
\begin{aligned}
    C_i\Pi_{\mathcal{S}}
    &=
    \frac{1}{|\mathcal{S}|}
    \sum_{G\in\mathcal{S}}C_iG \\
    &=
    \frac{1}{|\mathcal{S}|}
    \sum_{G\in\mathcal{S}}G \\
    &=
    \Pi_{\mathcal{S}}.
\end{aligned}
\]
Therefore every state in the image of $\Pi_{\mathcal{S}}$ is a
simultaneous $+1$ eigenstate of all the $C_i$.

It remains to show that this image is nonzero. Every nonidentity
Pauli operator has zero trace. Since $-I\notin\mathcal{S}$, the
identity is the only scalar element of $\mathcal{S}$, and therefore
\[
\begin{aligned}
    \operatorname{Tr}\Pi_{\mathcal{S}}
    &=
    \frac{1}{|\mathcal{S}|}
    \sum_{G\in\mathcal{S}}\operatorname{Tr}G \\
    &=
    \frac{\operatorname{Tr}I}{|\mathcal{S}|} \\
    &=
    \frac{2^n}{|\mathcal{S}|}
    >0.
\end{aligned}
\]
Because $\Pi_{\mathcal{S}}$ is a projector with positive trace, it has
nonzero rank and $H$ is frustration free.
\end{proof}

For commuting nonidentity Pauli terms, frustration requires an
incompatible signed relation. In a noncommuting family, anticommutation
can cause frustration even when the relation kernel is trivial. For example,
\begin{equation}
        H=-X-Z
        \label{eq:noncommuting-frustration-example}
\end{equation}
has $K_L=0$ but is frustrated because $X$ and $Z$ anticommute and cannot share a $+1$ eigenstate.

\subsection{Commuting relation-free Hamiltonian\label{sec: classicalsimu}}

\begin{property}[Polynomial filters of commuting relation-free
Hamiltonians are Clifford-equivalent to low-bond MPS]
\label{prop:commuting-source-mps}
Let
\begin{equation}
    H
    =
    \sum_{i=1}^{m}
    \lambda_i\gamma_i P_n(\ell_i),
    \qquad
    \lambda_i>0,
    \qquad
    \gamma_i\in\{\pm1\},
\end{equation}
be an $n$-qubit Pauli Hamiltonian whose indexed Pauli terms commute
pairwise. Let
\begin{equation}
    L
    =
    \begin{pmatrix}
        \ell_1 & \cdots & \ell_m
    \end{pmatrix}
\end{equation}
be its Pauli-label matrix, and suppose that
\begin{equation}
    K_L
    =
    \ker L
    =
    \{0\}.
\end{equation}
Let
\begin{equation}
    Q_D(x)
    =
    \sum_{k=0}^{D}q_kx^k
\end{equation}
be a real polynomial satisfying
\begin{equation}
    Q_D(H)\neq 0.
\end{equation}

Then there is an efficiently computable Clifford unitary $U$ such
that
\begin{equation}
    (U\otimes U^*)
    \left|\Psi_H^{Q_D}\right\rangle
    =
    |\varphi_D\rangle
    \otimes
    |\Phi_{n-m}\rangle,
\end{equation}
where, in the paired-site ordering
\begin{equation}
    (1,\overline 1),
    (2,\overline 2),
    \ldots,
    (m,\overline m),
\end{equation}
the state $|\varphi_D\rangle$ has an exact open-boundary MPS
representation of bond dimension at most
\begin{equation}
    D+1.
\end{equation}

\end{property}

\begin{proof}
The proof consists of four steps.

\medskip
\noindent\emph{Step 1: reduce the Hamiltonian to independent fields.}

Because the indexed Pauli terms commute,
\[
    L^T J_n L=0.
\]
Thus the subspace
\[
    W_L:=\operatorname{im}L
\]
is isotropic with respect to the binary symplectic form. Relation
freedom gives
\[
    \ker L=\{0\},
\]
and hence
\[
    \dim W_L
    =
    \operatorname{rank}L
    =
    m.
\]
An isotropic subspace of $\mathbb F_2^{2n}$ has dimension at most
$n$, so
\[
    m\leq n.
\]

By symplectic Gaussian elimination,
there is an efficiently computable Clifford unitary $U$ and signs
$\epsilon_i\in\{\pm1\}$ such that, after relabeling the qubits,
\[
    U P_n(\ell_i)U^\dagger
    =
    \epsilon_i Z_i,
    \qquad
    i\in[m].
\]
Absorbing these signs into the coefficients, define
\[
    a_i
    :=
    \epsilon_i\lambda_i\gamma_i.
\]
The transformed Hamiltonian is therefore
\[
    \widetilde H
    :=
    UHU^\dagger
    =
    H_{\mathrm{act}}\otimes I_{\mathrm{sp}},
\]
where
\[
    H_{\mathrm{act}}
    :=
    \sum_{i=1}^{m}a_iZ_i
\]
acts on the first $m$ qubits, and $I_{\mathrm{sp}}$ acts on the
remaining $n-m$ spectator qubits.

Polynomial functional calculus gives
\[
    Q_D(\widetilde H)
    =
    UQ_D(H)U^\dagger.
\]
Using the maximally entangled-state identity
\[
    (U^\dagger\otimes I)|\Phi_n\rangle
    =
    (I\otimes U^*)|\Phi_n\rangle,
\]
we obtain
\[
    \left|\Psi_{\widetilde H}^{Q_D}\right\rangle
    =
    (U\otimes U^*)
    \left|\Psi_H^{Q_D}\right\rangle.
\]
It is therefore sufficient to prove the MPS bound for
$\widetilde H$.

\medskip
\noindent\emph{Step 2: remove the spectator qubits.}

Since
\[
    \widetilde H
    =
    H_{\mathrm{act}}\otimes I_{\mathrm{sp}},
\]
one has
\[
    Q_D(\widetilde H)
    =
    Q_D(H_{\mathrm{act}})
    \otimes
    I_{\mathrm{sp}}.
\]
Also,
\[
    |\Phi_n\rangle
    =
    |\Phi_m\rangle
    \otimes
    |\Phi_{n-m}\rangle.
\]
Therefore the normalized filtered Choi state factorizes as
\[
    \left|\Psi_{\widetilde H}^{Q_D}\right\rangle
    =
    \left|\Psi_{H_{\mathrm{act}}}^{Q_D}\right\rangle
    \otimes
    |\Phi_{n-m}\rangle.
\]
The spectator Bell pairs have bond dimension one when each physical
qubit is grouped with its reference qubit. It remains only to bound
the bond dimension of
$\left|\Psi_{H_{\mathrm{act}}}^{Q_D}\right\rangle$.

\medskip
\noindent\emph{Step 3: bound the Schmidt rank across an arbitrary cut.}

Fix a cut after the $j$th active paired site. Define
\[
\begin{aligned}
    H_{\mathrm L}^{(j)}
    &:=
    \sum_{i=1}^{j}a_iZ_i,\\
    H_{\mathrm R}^{(j)}
    &:=
    \sum_{i=j+1}^{m}a_iZ_i.
\end{aligned}
\]
Across this bipartition,
\[
    H_{\mathrm{act}}
    =
    H_{\mathrm L}^{(j)}\otimes I_{\mathrm R}
    +
    I_{\mathrm L}\otimes H_{\mathrm R}^{(j)}.
\]
The two summands act on disjoint subsystems and commute. Hence the
binomial theorem gives
\[
\begin{aligned}
    Q_D(H_{\mathrm{act}})
    &=
    \sum_{k=0}^{D}
    q_k
    \left(
        H_{\mathrm L}^{(j)}\otimes I_{\mathrm R}
        +
        I_{\mathrm L}\otimes H_{\mathrm R}^{(j)}
    \right)^k
    \\
    &=
    \sum_{k=0}^{D}
    q_k
    \sum_{r=0}^{k}
    \binom{k}{r}
    \left(H_{\mathrm L}^{(j)}\right)^r
    \otimes
    \left(H_{\mathrm R}^{(j)}\right)^{k-r}.
\end{aligned}
\]
Regrouping the terms according to the left power $r$ gives
\[
    Q_D(H_{\mathrm{act}})
    =
    \sum_{r=0}^{D}
    \left(H_{\mathrm L}^{(j)}\right)^r
    \otimes
    B_r^{(j)},
\]
where
\[
    B_r^{(j)}
    :=
    \sum_{k=r}^{D}
    q_k
    \binom{k}{r}
    \left(H_{\mathrm R}^{(j)}\right)^{k-r}.
\]
Thus, across this cut, $Q_D(H_{\mathrm{act}})$ is a sum of at most
$D+1$ left-right product operators.

For an operator $A$, define its Choi vector by
\[
    |A\rangle\rangle
    :=
    (A\otimes I)|\Phi\rangle.
\]
Vectorization turns an operator tensor product into a state tensor
product:
\[
    |A_{\mathrm L}\otimes A_{\mathrm R}\rangle\rangle
    =
    |A_{\mathrm L}\rangle\rangle
    \otimes
    |A_{\mathrm R}\rangle\rangle.
\]
Therefore, the unnormalized filtered Choi vector satisfies
\[
    \left|
        Q_D(H_{\mathrm{act}})
    \right\rangle\rangle
    =
    \sum_{r=0}^{D}
    \left|
        \left(H_{\mathrm L}^{(j)}\right)^r
    \right\rangle\rangle
    \otimes
    \left|
        B_r^{(j)}
    \right\rangle\rangle.
\]
This is a decomposition into at most $D+1$ product vectors across
the doubled cut between $\mathrm L\overline{\mathrm L}$ and $\mathrm R\overline{\mathrm R}$. Consequently,
\[
    \operatorname{SchmidtRank}_{\mathrm L\overline{\mathrm L}
    :
    \mathrm R\overline{\mathrm R}}
    \left(
        \left|\Psi_{H_{\mathrm{act}}}^{Q_D}\right\rangle
    \right)
    \leq
    D+1.
\]
The cut $j$ was arbitrary, so this bound holds across every
consecutive cut in the paired-site ordering.

Successive Schmidt decompositions from left to right therefore give
an exact open-boundary MPS whose bond dimension at every cut is at
most
\[
    D+1.
\]

\medskip
\noindent\emph{Step 4: explicit tensors and preparation access.}
The rank bound can be made constructive. For $0\le r,s\le D$, define
\[
 B(t)_{rs}:=\begin{cases}\binom{s}{r}t^{s-r},&r\le s,\\0,&r>s,\end{cases}
 \qquad \mathbf q:=(q_0,\ldots,q_D)^T,
 \qquad \mathbf e_0:=(1,0,\ldots,0)^T.
\]
The identity $(1,p,\ldots,p^D)B(t)=(1,p+t,\ldots,(p+t)^D)$ gives
\[
 \mathbf e_0^TB(t_1)\cdots B(t_m)\mathbf q
 =Q_D(t_1+\cdots+t_m).
\]
Consequently the paired-site tensors
\[
 A_i^{bb'}:=2^{-1/2}\delta_{b,b'}B(a_i(-1)^b),
 \qquad b,b'\in\{0,1\},
\]
with boundaries $\mathbf e_0^T$ and $\mathbf q$, represent exactly
$(Q_D(H_{\rm act})\otimes I)|\Phi_m\rangle$ and have bond dimension $D+1$.
They require only polynomially many arithmetic operations in $m,D$.
\end{proof}

\end{document}